\documentclass{amsart}
\usepackage[margin=4cm]{geometry}
\usepackage{amsmath, amsfonts, amsthm, amssymb, mathtools}
\usepackage{mathrsfs}
\usepackage{subfig}
\usepackage{stmaryrd}
\usepackage{verbatim}
\usepackage{hyperref}
\usepackage{xcolor}
\usepackage{graphicx}
\usepackage{algorithmic, algorithm}
\usepackage{enumerate}
\numberwithin{equation}{section}

\usepackage[textsize=tiny]{todonotes}

\mathtoolsset{showonlyrefs=true}
\newtheorem{theorem}{Theorem}[section]
\newtheorem{corollary}[theorem]{Corollary}
\newtheorem{lemma}[theorem]{Lemma}
\newtheorem{proposition}[theorem]{Proposition}

\theoremstyle{definition}
\newtheorem{definition}[theorem]{Definition}
\newtheorem{remark}[theorem]{Remark}

\newtheorem{assumption}[theorem]{Assumption}

\newcommand{\ind}{1\hspace{-2.1mm}{1}} 
\newcommand{\RR}{\mathbb{R}}
\newcommand{\NN}{\mathbb{N}}
\newcommand{\PP}{\mathbb{P}}
\newcommand{\QQ}{\mathbb{Q}}

\newcommand{\EE}{\mathbb{E}}
\newcommand{\FF}{\mathbb{F}}
\newcommand{\XX}{\mathbf{X}}
\newcommand{\xx}{\mathbf{x}}
\newcommand{\yy}{\mathbf{y}}
\newcommand{\zz}{\mathbf{z}}
\newcommand{\vv}{\mathbf{v}}
\newcommand{\Aa}{\mathcal{A}}
\newcommand{\Bb}{\mathcal{B}}
\newcommand{\Cc}{\mathcal{C}}
\newcommand{\Ff}{\mathcal{F}}

\newcommand{\Nn}{\mathcal{N}}
\newcommand{\Oo}{\mathcal{O}}
\newcommand{\Ll}{\mathcal{L}}

\newcommand{\Ss}{\mathcal{S}}
\newcommand{\Qq}{\mathcal{Q}}
\newcommand{\Uu}{\mathcal{U}}
\newcommand{\Vv}{\mathcal{V}}
\newcommand{\Ww}{\mathcal{W}}
\newcommand{\Zz}{\mathcal{Z}}

\newcommand{\D}{\mathrm{d}}
\newcommand{\Dx}{\mathrm{D}_{x}}
\newcommand{\Dxx}{\mathrm{D}_{\xx}}
\newcommand{\Dr}{\mathrm{D}}
\newcommand{\E}{\mathrm{e}}
\newcommand{\I}{\mathrm{i}}
\newcommand{\X}{\mathrm{X}}
\newcommand{\Y}{\mathrm{Y}}

\newcommand{\eps}{\varepsilon}

\newcommand{\Uf}{\mathfrak{U}}
\newcommand{\Zf}{\mathfrak{Z}}

\newcommand{\RWNN}{\aleph}

\newcommand{\varrhob}{\boldsymbol\varrho}

\DeclareMathOperator{\relu}{\varsigma}
\DeclareMathOperator{\brelu}{\boldsymbol\varsigma}

\DeclareMathOperator*{\argmin}{argmin}
\newcommand{\Hb}{\boldsymbol{H}}
\newcommand{\floor}[1]{\lfloor #1 \rfloor}

\newcommand{\phib}{\boldsymbol{\phi}}
\newcommand{\Sb}{\boldsymbol{S}}
\newcommand{\Am}{\mathrm{A}}
\newcommand{\Imat}{\mathrm{I}}
\newcommand{\bm}{\mathrm{b}}
\newcommand{\betam}{\boldsymbol{\beta}}
\newcommand{\Ee}{\mathcal{E}}
\newcommand{\half}{\frac{1}{2}}
\newcommand{\Rr}{\mathcal{R}}
\newcommand{\Df}{\mathfrak{D}}
\newcommand{\Ef}{\mathfrak{E}}
\newcommand{\If}{\mathfrak{I}}

\newcommand{\ErrYiVhat}{\Ef\left[Y_{t_i}, \widehat{\Vv}_{t_i}\right]}
\newcommand{\ErrVhatiUfi}{\Ef\left[\widehat{\Vv}_{t_i}, \widehat{\Uf}_{t_i}\right]}
\newcommand{\ErrIntYtVihat}{\If_i\left[Y_t, \widehat{\Vv}_{t_i}\right]}
\newcommand{\ErrIntZkihat}{\If_i\left[Z^k, \Zowk_{t_i}\right]}
\newcommand{\ErrIntZkibar}{\If_i\left[Z^k, \overline{Z}^{k}_{t_i}\right]}
\newcommand{\ErrZkibarhat}{\Ef\left[\overline{Z}^k_{t_i}, \Zowk_{t_i}\right]}
\newcommand{\ErrYiUfi}{\Ef\left[Y_{t_i}, \widehat{\Uf}_{i}\right]}
\newcommand{\ErriYiUfip}{\Df_i\left[Y_{t_{i+1}}, \widehat{\Uf}_{i+1}\right]}
\newcommand{\ErrYiUfip}{\Ef\left[Y_{t_{i+1}}, \widehat{\Uf}_{i+1}\right]}

\newcommand{\ZowO}{\overline{\widehat{Z}}\kern0.1em^{1}}
\newcommand{\ZowT}{\overline{\widehat{Z}}\kern0.1em^{2}}
\newcommand{\Zowk}{\overline{\widehat{Z}}\kern0.1em^{k}}

\newcommand{\zowO}{\overline{\widehat{z}}^{1}}
\newcommand{\zowT}{\overline{\widehat{z}}^{2}}
\newcommand{\zowk}{\overline{\widehat{z}}^{k}}

\newcommand{\Deli}{\delta_{i}}
\newcommand{\DelWi}{\Delta^{W}_{i}}

\newcommand{\aplim}{\mathrm{ap}\lim}

\title[Random neural networks for rough volatility]{Random neural networks for rough volatility}

\author{Antoine Jacquier}
\address{Department of Mathematics, Imperial College London, and the Alan Turing Institute}
\email{a.jacquier@imperial.ac.uk}

\author{\v{Z}an \v{Z}uri\v{c}}
\address{Kaiju Capital Management}
\email{z.zuric19@imperial.ac.uk}

\thanks{AJ acknowledges financial support from the EPSRC/T032146 grant. \v{Z}\v{Z} was supported by the EPSRC/S023925 CDT in Mathematics of Random Systems.
We would like to thank Lukas Gonon, Christian Bayer and Jinniao Qiu for helpful discussions.
The \texttt{python} code is available at \href{https://github.com/ZuricZ/RWNN_PDE_solver}{\texttt{ZuricZ/RWNN\_PDE\_solver}}.
The authors have no relevant financial or non-financial interests to disclose.}

\begin{document}
\maketitle

\begin{abstract}
We construct a deep learning-based numerical algorithm to solve path-dependent partial differential equations arising in the context of rough volatility.
Our approach is based on interpreting the PDE as a solution to an BSDE, building upon recent insights by Bayer, Qiu and Yao, 
and on constructing a neural network of reservoir type as originally developed by Gonon, Grigoryeva, Ortega.
The reservoir approach allows us to formulate the optimisation problem as a simple least-square regression for which we prove theoretical convergence properties.
\end{abstract}

\tableofcontents

\section{Introduction}

In recent years, a fundamental shift from classical modelling towards so-called rough stochastic volatility models has happened. These ``rough" models were first proposed by Gatheral, Jusselin, Rosenbaum~\cite{Gatheral2018} and by Bayer, Gatheral, Friz~\cite{Bayer2015},
and have since sparked a great deal of research, because of their ability to capture stylised facts of volatility time series and of option prices more accurately, while remaining parsimonious. In essence, they are a class of continuous-path stochastic volatility models, where the instantaneous volatility is driven by a stochastic process with paths rougher than those of Brownian Motion, typically modelled by a fractional Brownian motion~\cite{Mandelbrot1968} with Hurst parameter $H\in(0,1)$. 
The reason for this drastic paradigm shift can be found not only under the historical measure, where the roughness of the time series of daily log-realised variance estimates suggests H{\"o}lder regularity of $H\approx 0.1$, but also under the pricing measure, where rough volatility models are able to reproduce the power-law behaviour of the ATM volatility skew. Since then, a slew of papers have appeared, providing closed-form expressions for the characteristic functions of rough Heston models~\cite{Euch2018}, machine learning techniques for calibration~\cite{Horvath2020}, microstructural foundations~\cite{Euch2018a}, option pricing partial differential equations (PDEs) solvers~\cite{jacquier2019deep, Bayer2022}, among others.
A full overview can be found in the recent monograph~\cite{bayer2023rough}.

Dating back to Black-Scholes~\cite{Black1973}, PDEs have been used to model the evolution of the prices of European-style options. 
However, rough volatility models give rise to a non-Markovian framework, where the value function for a European option is not deterministic anymore, but is instead random and satisfies a backward stochastic partial differential equation (BSPDE) as was shown in~\cite{Bayer2022}.
Moreover, even in classical diffusive models, the so-called curse of dimensionality poses a challenge when solving PDEs in high dimension;
until recently, only the backward stochastic differential equation (BSDE) approach by~\cite{Pardoux1990} was available to tackle this, 
which is not really feasible in dimension beyond six.

On a positive note, machine learning methods have spread inside quantitative finance in recent years, and neural networks, in particular, have become a powerful tool to overcome problems in high-dimensional situations, because of their superior computational performance across a wide range of applications~\cite{Buehler2019, gierjatowicz2022robust, Ruf2020};
more precisely in the context of PDEs, examples of applications thereof can be found in~\cite{Weinan2017, Han2018, Sirignano2018, jacquier2019deep, Saporito2021, Beck2021}. 
For a more thorough literature review on the use of neural networks in finance and finance-related PDEs, we refer the reader to the surveys in~\cite{beck2020overview, germain2021neural}.

In this paper, we focus on the works by Hur{\'e}, Pham and Warin~\cite{Hure2020}, and by Bayer, Qiu and Yao~\cite{Bayer2022}, where the classical backward resolution technique is combined with neural networks to estimate both the value function and its gradient. 
Not only does this approach successfully reduce the curse of dimensionality, but also appears more effective in both accuracy and computational efficiency than existing Euler-based approaches.

Besides research on numerical aspects, a lot of progress has been made on the theoretical foundations for neural network-based methods, 
in particular showing that they are able to approximate solutions of certain types of PDEs~\cite{Elbraechter2021, Hutzenthaler2020, Reisinger2020, gonon2021deep}. 
These results are significant as they show that deep neural networks can be used to solve complex problems that were previously thought intractable. 
However, in practice, optimal parameters of any given neural network minimising a loss function ultimately have to be calculated approximately. This is usually done through some kind of stochastic gradient descent (SGD) algorithm, which inadvertently introduces an optimisation error. Because of the non-convexity of the network's loss surface and the stochastic nature of the SGD, the optimisation error is notoriously hard to treat rigorously. 
One such attempt by Gonon~\cite{gonon2021random} instead involves the use of neural networks in which only certain weights are trainable and the remaining are randomly fixed. 
This suggests that these random-weight neural networks are, in fact, capable of learning non-degenerate Black-Scholes-type PDEs without succumbing to the curse of dimensionality. Following this, we combine the classical BSDE approach~\cite{Pardoux1990, Briand2003} with random-weight neural networks (RWNNs)~\cite{Huang2006, rahimi2007random, rahimi2008weighted}.

Our final algorithm then reduces to a least-square Monte-Carlo, introduced by Longstaff and Schwartz~\cite{Longstaff2001} (see also~\cite{Anker2017} for related applications), where the usually arbitrary choice of basis is `outsourced' to the reservoir of the corresponding RWNN. The basis is computationally efficient and ultimately allows us to express the approximation error in terms of the number of stochastic nodes in the network. Moreover, vectorisation of the Randomised Least Square along the sampling direction allows us to evaluate the sum of outer tensor products using the \texttt{einsum} function (available in \texttt{NumPy} and \texttt{PyTorch}) and achieve an even greater speed-up. 
{One word of caution though: in our numerical examples, for the rough Bergomi model and for Basket options, computation time is still longer than using Monte Carlo methods, mostly because it does require simulating sample paths.
It however opens the gates to (more advanced) numerical schemes for path-dependent partial differential equations,
which we plan to investigate more in later projects.
}

To summarise, in contrast with Bayer-Qiu-Yao~\cite{Bayer2022}, our numerical scheme employs RWNNs as opposed to the conventional feed-forward neural networks, resulting in significantly faster training times without sacrificing the accuracy of the scheme. 
Moreover, this structure allows us to provide error bounds in terms of the number of hidden nodes, granting additional insights into the network's performance. Given the comparable performance of RWNNs and conventional feed-forward neural networks, we argue that this paper illuminates an essential lesson, namely that the additional complexity of deep neural networks can sometimes be redundant at the cost of precise error bounds.
We note in passing that RWNNs have already been used in Finance to price American options~\cite{herrera2021optimal}, 
for financial data forecasting~\cite{Liu2018}, for PIDEs~\cite{gonon2021deep},
and we refer the interested reader to~\cite{Cao2018} for a general overview of their applications in data science.

{Moreover, in parallel to our work, Shang, Wang, and Sun \cite{Shang2023} combined randomised neural networks with Petrov-Galerkin methods to solve linear and non-linear PDEs. Their method, similar to ours, uses randomly initialised neural networks with trainable linear readouts. Neufeld, Schmocker, and Wu \cite{Neufeld2025} conducted a comprehensive error analysis of the random deep splitting method for non-linear parabolic PDEs and PIDEs, demonstrating high-dimensional problem-solving capabilities.}

The paper is structured as follows:
Section~\ref{sec:RWNN} provides a brief overview of Random-weight Neural Networks (RWNNs), including their key features and characteristics. 
In Section~\ref{sec:markovian_case}, we outline the scheme for the Markovian case and discuss the non-Markovian case in Section~\ref{sec:non-Markovian_case}.
The convergence analysis is presented in Section~\ref{sec:convergence_analysis}. Additionally, Section~\ref{sec:RWNN_numerical_results} presents numerical results, which highlight the practical relevance of the scheme and its performance for different models.
Some of the technical proofs are postponed to Appendix~\ref{apx:technical_proofs} to ease the flow of the paper.

\textbf{Notations:}
$\RR^+ = [0,\infty)$;
$\RWNN$ refers to a random neural network,
defined in Section~\ref{sec:RWNN};
for an open subset $E\subset \RR^d$, 
$1\leq p \leq \infty$ and $s \in \NN$ 
we define the Sobolev space
$$
\Ww^{s, p}(E, \RR^m)
\coloneqq  \Big\{f \in L^p(E, \RR^m): \; \partial_{\xx}^{\boldsymbol{\alpha}} f \in L^p(E, \RR^m), \text{for all } |\boldsymbol{\alpha}| \leq s\Big\},
$$
where $\boldsymbol{\alpha} = \left(\alpha_1, \ldots, \alpha_d\right)$, $|\boldsymbol{\alpha}|=\alpha_1+\ldots+\alpha_d$, and the derivatives $\partial_{\xx}^{\boldsymbol{\alpha}} f = \partial_{x_1}^{\alpha_1} \dots \partial_{x_d}^{\alpha_d} f$ are taken in a weak sense.
{To be consistent with probabilistic notations---although Machine Learning literature tends to differ---we shall write $\EE^\Phi[\cdot] := \EE[\cdot\vert\Phi]$ as the conditional expectation with respect to the random variable~$\Phi$.}

\section{Random-weight neural network (RWNN)}\label{sec:RWNN}

Neural networks with random weights appeared in the seminal works by Barron~\cite{Barron1993,barron1992neural}, 
but a modern version was proposed by Huang~\cite{Huang2006} under the name \textit{Extreme learning machine},
and today are known under different names: reservoir networks, random feature or random-weight networks; we adopt the latter as it sounds more explicit to us. 

\begin{definition}[Neural network]\label{def:neuralnet}
Let $L, N_{0}, \ldots, N_{L} \in \NN, \varrho: \RR \rightarrow \RR$ and for $l=$
$1, \ldots, L$ let $w_{l}: \RR^{N_{l-1}} \rightarrow \RR^{N_{l}}$ an affine function. A function $F: \RR^{N_{0}} \rightarrow \RR^{N_{L}}$
defined as
\[
F=w_{L} \circ F_{L-1} \circ \cdots \circ F_{1}, 
\quad \text {with } F_{l}=\varrho \circ w_{l} \quad \text { for } l=1, \ldots, L-1,
\]
is called a \textit{neural network},
with activation function~$\varrho$
applied component-wise. $L$ denotes the total number of layers, $N_{1}, \ldots, N_{L-1}$ denote
the dimensions of the hidden layers and~$N_{0}$ and~$N_{L}$ those of the input and output layers respectively. 
For each $l\in\{1, \dots, L\}$ the affine function $w_{l}:\RR^{N_{l-1}}\to\RR^{N_{l}}$ is given as $w_{l}(\xx) = \Am^{(l)} \xx + \bm^{(l)}$,
for $\xx\in \RR^{N_{l-1}}$, with $\Am^{(l)} \in \RR^{N_{l}\times N_{l-1}} $ and $\bm^{(l)} \in \RR^{N_{l}}$.
For any $i\in\{1, \dots N_{l}\}$ and $j\in\{1, \dots, N_{l-1}\}$, 
$A_{i j}^{(l)}$ is interpreted as the weight of the edge connecting node~$i$ of layer~$l-1$ to node~$j$ of layer~$l$.
\end{definition}

A \textit{random-weight neural network} (RWNN) is a neural network where the hidden layers are randomly sampled from a given distribution and then fixed; 
consequently, only the last layer is trained:
out of all the parameters $(\Am^{(l)}, \bm^{(l)})_{l=0,\ldots,L}$ of the $L$-layered neural network, the parameters $(\Am^{(0)}, \bm^{(0)}, \ldots, A^{(L-1)}, \bm^{(L-1)})$ are randomly sampled and frozen and only $(\Am^{(L)}, \bm^{(L)})$ from the last layer are trained. 

The training of such an RWNN can then be simplified into a convex optimisation problem. This makes the training easier to manage and understand both practically and theoretically. However, by only allowing certain parts of the parameters to be trained, the overall capacity and expressivity are possibly reduced. Although it is still unclear if random neural networks still maintain any of the powerful approximation properties of general deep neural networks, these questions have been addressed to some extent in e.g.~\cite{gonon2023approximation, Mei2022}, where learning error bounds for RWNNs have been proved.

Denote now $\RWNN_{\infty}^{\varrho}(d_0,d_1)$ the set of random neural networks from $\RR^{d_0}$ to $\RR^{d_1}$, with activation function~$\varrho$--and we shall drop the explicit reference to input and output dimensions in the notation whenever they are clear from the context. 
Moreover, for any $L, K \in \NN$, $\RWNN_{L, K}^\varrho$ represents a random neural network with a fixed number of hidden layers $L$ and fixed input and output dimension $K$ for each hidden layer. 
We now give a precise definition of a single layer $\RWNN^\varrho_{K}\coloneqq\RWNN^\varrho_{1, K}$, which we will use for our approximation.

\begin{definition}[Single layer RWNN]\label{def:RWNN}
Let $(\widetilde{\Omega}, \widetilde{\Ff}, \widetilde{\PP})$ be a probability space on which the iid random variables on a bounded domain $\Am_{k}: \widetilde{\Omega} \rightarrow \Ss \subset \RR^{d}$ and $b_{k}: \widetilde{\Omega} \rightarrow \mathscr{sS} \subset \RR$, respectively corresponding to weights and biases, are defined. Let $\phib=\left\{\phi_{k}\right\}_{k \geq 1}$ denote a sequence of random basis functions, where each $\phi_{k}: \RR^{d} \to\RR$ is of the form
$$
\phi_{k}(\xx)\coloneqq \varrho\left(\Am_{k}^{\top} \xx + b_{k}\right), \qquad x\in\RR^d,
$$
with $\varrho:\RR\rightarrow\RR$ a Lipschitz continuous activation function. 
For an output dimension~$m$ and~$K$ hidden units, we define the \textit{reservoir} or \textit{random basis} as $\Phi_{K}\coloneqq\phi_{1:K}=(\phi_1,\dots,\phi_K)$ and the random network $\RWNN^\varrho_K$ with parameter $\Theta=\left(\theta_{1}, \ldots, \theta_{m}\right)^\top \in \RR^{m\times K}$ as the map
\[
\RWNN^\varrho_K: \xx \mapsto \Psi_K(\xx;\Theta)\coloneqq\Theta\Phi_K(\xx).
\]
Thus, for each output dimension $j\in\{1,\dots,m\}$, $\RWNN^\varrho_K$ produces a linear combination of the first $K$ random basis functions $\theta_j^\top\phi_{1:K}\coloneqq \sum_{k=1}^{K} \theta_{j,k} \phi_{k}$.
\end{definition}

\begin{remark}
In this paper, we will make use of the more compact vector notation
$$
\Phi_K: \RR^d \ni \xx\mapsto \varrhob(\Am \xx+\bm) \in \RR^{K},
$$
where $\varrhob:\RR^K\rightarrow\RR^K$ acts component-wise $\varrhob(\yy)\coloneqq(\varrho(y_1), \dots \varrho(y_K))$ and $\Am: \widetilde{\Omega} \rightarrow \RR^{K\times d}$ and $\bm: \widetilde{\Omega} \rightarrow \RR^K$ are the random matrix and bias respectively.
\end{remark}

\subsection{Derivatives of ReLu-RWNN}
\label{sec:RWNNderivatives}
In recent years ReLu neural networks have been predominately used in deep learning, because of their simplicity, efficiency and ability to address the so-called vanishing gradient problem~\cite{LeCun2012}. In many ways, ReLu networks also give a more tractable structure to the optimisation problem compared to their smooth counterparts such as $\tanh$ and sigmoid. 
Gonon, Grigoryeva and Ortega~\cite{gonon2023approximation} derived error bounds to the convergence of a single layer RWNN with ReLu activations. 
Now, while ${\relu(y)\coloneqq \max\{y,0\}}$ is performing well numerically, it is, however, not differentiable at zero (see~\cite{Berner2019} for a short exposition on the chain-rule in ReLu networks). As ReLu-RWNNs will be used in our approach to approximate solutions of partial differential equations, a discussion on its derivatives is in order. 
To that end we let $\brelu(\yy)\coloneqq(\relu(y_1),\dots,\relu(y_K))$ and $\Hb(y)=\ind_{(0,\infty)}(\yy)\in\RR^K$ for $\yy\in\RR^K$, where the indicator function is again applied component-wise.

\begin{lemma}\label{lem:RWNNderivative}
For any linear function $\ell(\xx)=\Am \xx + \bm$, with $\Am\in\RR^{K\times d}$ and $\bm\in\RR^K$, then
\[
\nabla_x(\brelu\circ \ell)(\xx) = \operatorname{diag}(\Hb(\Am \xx+\bm))\Am,
\qquad \text{for a.e. } \xx\in\RR^d.
\]
\end{lemma}
\begin{proof}
Let $\Aa\coloneqq\left\{\xx\in\RR^d: (\brelu\circ\ell)(\xx)=0\right\}
=\left\{\xx\in\RR^d:\ell(\xx)\leq 0\right\}$. 
Then ${(\brelu\circ\ell)(\xx)=\ell(\xx)}$ for all $\xx\in\RR^d\setminus\Aa$.
Since~$\ell$ is Lipschitz,
differentiability on level sets~\cite[Section~3.1.2,~Corollary~I]{Evans1992} implies that
$\nabla_{\xx}\left(\brelu\circ\ell\right)(\xx) = \boldsymbol{0}\in\RR^d$
for almost every $\xx\in\Aa$,
and hence
\[
\nabla_{\xx}(\brelu\circ\ell)(\xx)
=\operatorname{diag}\left(\ind_{\{\ell(\xx)\in\RR^d\setminus\Aa\}}\right)\nabla_{\xx}\ell(\xx)
=\operatorname{diag}\left(\ind_{(0,\infty)}(\ell(\xx))\right)\nabla_{\xx}\ell(\xx)
=\operatorname{diag}(\Hb(\Am \xx+\bm))\Am.
\]
\end{proof}
Thus by Lemma~\ref{lem:RWNNderivative}, the first derivative of $\Psi(\cdot;\Theta)\in\RWNN^{\relu}_K$ is equal to 
\begin{equation}\label{eq:RWNN1stderivative}
    \nabla_{\xx}\Psi_K(\xx; \Theta)
    = \Theta\operatorname{diag}(\Hb(\Am \xx+\bm))\Am
    \qquad \text{for a.e. } \xx\in\RR^d.
\end{equation}
The above statements hold almost everywhere, it is thus appropriate we introduce a notion of \textit{approximate differentiability}.
\begin{definition}[{Approximate~limit,~\cite[Section~1.7.2]{Evans1992}}]
Consider a Lebesgue-measurable set $E \subset \RR^d$, a measurable function $f: E \rightarrow \RR^m$ and a point $\xx_0 \in E$. We say $l \in \RR^m$ is the approximate limit of~$f$ at~$\xx_0$, and write
$\aplim_{\xx \rightarrow x_0} f(\xx)=l$,
if for each $\eps>0$,
\[
\lim _{r \downarrow 0} \frac{\lambda\left(\Bb_r(\xx_0) \cap\left\{\xx\in E:\;|f(\xx)-l| \geq \eps\right\}\right)}{\lambda(\Bb_r(\xx_0))}=0,
\]
with~$\lambda$ the Lebesgue measure and $\Bb_r(\xx_0)$ the closed ball with radius $r>0$ and center~$\xx_0$.
\end{definition}
\begin{definition}[{Approximate differentiability,~\cite[Section~6.1.3]{Evans1992}}]\label{def:approximatediff}
Consider a measurable set $E\subset \RR^d$, a measurable map $f:E\rightarrow \RR^m$ and a point $\xx_0\in E$. 
The map~$f$ is approximately differentiable at~$\xx_0$ if there exists a linear map $\Dxx:\RR^d\rightarrow\RR^m$ such that
\[
\aplim_{\xx \rightarrow x_{0}} \frac{f(\xx)-f(\xx_{0})-\Dxx(\xx-\xx_{0})}{|\xx-\xx_{0}|}=0.
\]
Then~$\Dxx$ is called the approximate differential of~$f$ at~$\xx_0$.
{We call~$f$ approximately differentiable almost everywhere if its approximately derivative exists almost everywhere.}
\end{definition}
\begin{remark}\label{rem:usualdiffrules}
The usual rules from classical derivatives, such as the uniqueness of the differential, and differentiability of sums, products and quotients, apply to approximately differentiable functions. 
Moreover, the chain rule applies to compositions $\varphi\circ f$ when~$f$ is approximately differentiable at~$\xx_0$ and $\varphi$ is classically differentiable at $f(\xx_0)$.
\end{remark}
\begin{remark}[{\cite[Theorem~4, Section~6.1.3]{Evans1992}}]\label{rem:approxweakeq}
For $f\in \Ww_{\mathrm{loc}}^{1, p}\left(\mathbb{R}^{d}\right)$ and $1 \leq p \leq \infty$, $f$~is approximately differentiable almost everywhere and its approximate derivative equals its weak derivative almost everywhere. 
We will thus use the operator~$\Dxx$ to denote the weak derivative and approximate derivative interchangeably, to distinguish them from the classical derivative denoted by $\nabla$.
\end{remark}

\begin{lemma}\label{lem:diffexpectationeq}
Let $E\subset \RR^d$ be a measurable set with finite measure, $X:\Omega\rightarrow E$ a continuous random variable on some probability space $(\Omega, \Ff, \PP)$,
$\varphi\in\Cc^1(\RR^m)$, $\Phi_{\mathrm{ap}}: E\rightarrow \RR^m$ an  approximately differentiable function, and~$\Phi$ its $\Cc^1(\RR^d; \RR^m)$ extension to $\RR^d$.
Then
$\EE[\varphi(\Dx\Phi_{\mathrm{ap}}(X))] = \EE[\varphi(\nabla_x\Phi(X))]
$.

\end{lemma}
\begin{proof}
By~\cite[Theorem 3.1.6]{Federer1996} a function $\Phi_{\mathrm{ap}}: E \rightarrow \mathbb{R}^{m}$ is approximately differentiable almost everywhere if for every $\eps>0$ there is a compact set $F \subset E$ such that the Lebesgue measure $\lambda(E \backslash F)<\eps$,  $\left.\Phi_{\mathrm{ap}}\right|_{F}$ is $\Cc^{1}$ {and there exists a~$\Cc^1$-extension on~$\RR^d$}. 
Since~$\varphi$ is everywhere differentiable, 
it maps null-sets to null-sets~\cite[Lemma 7.25]{real1986complex}. 
The claim follows since~$\PP$ is absolutely continuous with respect to the Lebesgue measure~$\lambda$, $X$~being a continuous random variable.

\end{proof}
\begin{corollary}\label{coro:diffexpectationeq2}
    Let $E\subset \RR^d$ be a measurable set with finite measure, $X:\Omega\rightarrow E$ a continuous random variable on some probability space $(\Omega, \Ff, \PP)$, $\varphi\in\Cc^1(\RR^m)$, $\Phi:E\rightarrow \RR^m$  an approximately differentiable function, 
    and $\Psi\in \Ww^{1,p}(E,\RR^m)$
    for $p\geq 1$ such that $\Phi = \Psi$ almost everywhere. Then
    $\EE[\varphi(\Dxx\Phi(X))] = \EE[\varphi(\Dxx\Psi(X))]$.
\end{corollary}
\begin{proof}
This is a direct consequence of Lemma~\ref{lem:diffexpectationeq}, after noting that the two notions of derivatives are the same on $\Ww^{1,p}(E,\RR^m)$ (see Remark~\ref{rem:approxweakeq}).
\end{proof}

From a practical perspective, the second-order derivative of the network with respect to the input will be zero for all intents and purposes. 
However, as will become apparent in Lemma~\ref{lem:derivativeBounds}, we need to investigate it further, in particular the measure zero set of points where ReLu-RWNN is not differentiable. Rewriting the diagonal operator in terms of the natural basis $\{e_{i}\}$ and evaluating the function $\Hb$ component-wise yields
$$
\nabla_{\xx}\Psi_K(\xx; \Theta)=\Theta \left(\sum_{j=1}^{K} e_{j} e_{j}^\top H\left(e_{j}^\top \Am \xx+b_{j}\right) \right) \Am.
$$
The $i$-th component of the second derivative is thus
$$
\left[\nabla_{\xx}^2\Psi_K(\xx; \Theta)\right]_i
= \Theta\left( \sum_{j=1}^{K} e_{j} e_{j}^\top a_{j i}H'\left(e_{j}^\top \Am \xx+b_{j}\right)\right)\Am \\
    = \Theta\operatorname{diag}\left(a_{i}\right)\operatorname{diag}\left(\Hb'(\Am \xx+\bm)\right) \Am,
$$
where $a_i$ denotes the $i$-th column of the matrix~$\Am$. 
Next, we let
$$\delta_0^{\eps}(\xx) \coloneqq \frac{H(\xx)-H(\xx-\eps)}{\eps}$$
for $\xx\in\RR$ and define the left derivative of $H$ as $H'=\lim_{\eps\downarrow 0}\delta_0^\eps=\delta_0$ in the distributional sense. This finally gives the second derivative of the network:
\begin{equation}\label{eq:SecondDiff}
\left[\nabla_{\xx}^2\Psi_K(\xx; \Theta)\right]_i
 = \Theta\operatorname{diag}\left(a_{i}\right)\operatorname{diag}\left(\boldsymbol{\delta}_0(\Am \xx+\bm)\right) \Am,
\end{equation}
where $\boldsymbol{\delta}_0$ denotes the vector function applying $\delta_0$ component-wise.

\subsection{Randomised least squares (RLS)}
\label{sec:randomizedLS}
Let $Y\in\RR^d$ and $X\in\RR^k$ random variables on some probability space $(\Omega, \Ff, \PP)$ and $\betam\in\RR^{d\times k}$ a deterministic matrix. 
If the loss function is the mean square error (MSE), the randomised least square estimator reads
\begin{align*}
    \nabla_{\betam} \EE\left[\|Y-\betam X\|^2\right] &= \nabla_{\betam} \EE[(Y-\betam X)^\top(Y-\betam X)] \\ 
    &= \EE\left[\nabla_{\betam} (Y^\top Y - Y^\top \betam X - X^\top\betam^\top Y + X^\top\betam^\top\betam X)\right] \\
    &= \EE\left[2\betam XX^\top - 2YX^\top\right],
\end{align*}
which gives the minimiser\footnote{The matrix $\EE[XX^\top]$ may not be invertible, but its generalised Moore-Penrose inverse always exists.}
$\betam = \EE\left[YX^\top\right] \EE\left[XX^\top\right]^{-1}$,
and its estimator
\begin{equation}\label{eq:RLS_estimator}
\widehat\betam \coloneqq  \left(\sum^n_{j=1} Y_j X_{j}^{\top}\right)\left(\sum^n_{j=1} X_jX_j^\top\right)^{-1}.
\end{equation}
Depending on the realisation of the reservoir of the RWNN, the covariates of~$X$ may be collinear, 
so that~$X$ is close to rank deficient. 
A standard remedy is to use the Ridge regularised version~\cite{Hoerl1970} of the estimator
\[
\widehat\betam_R = \left(\sum^n_{j=1} Y_j X_{j}^{\top}\right)\left(\sum^n_{j=1} X_jX_j^\top+\lambda I\right)^{-1},
\quad \text{for } \lambda>0,
\]
which results in a superior, more robust performance in our experiments.
\begin{remark}
The above derivation holds true for the approximate derivative~$\Dx$ as well because all operations above hold for approximately differentiable functions (Remark~\ref{rem:usualdiffrules}).
\end{remark}
\begin{remark}
{At first glance, the form of the RLS estimator in~\eqref{eq:RLS_estimator} suggests that the sum of outer products over~$n>>1$ samples may be computationally expensive. In practice, however, this operation can be implemented efficiently by exploiting the tensor functionalities provided by libraries such as \texttt{NumPy} and \texttt{PyTorch}. In particular, the \texttt{einsum} function enables an efficient evaluation of the required sum of outer products, thereby further optimising the overall computation. Implementation details are provided in the accompanying code, available at \href{https://github.com/ZuricZ/RWNN_PDE_solver}{\texttt{ZuricZ/RWNN\_PDE\_solver}}.}

\end{remark}

\section{The Markovian case}\label{sec:markovian_case}
Let the process $\XX$ of the traded and non-traded components of the underlying under the risk-neutral measure $\QQ$ be given by the following $d$-dimensional SDE:
\begin{equation}\label{eq:markovianSDE}
\XX_s^{t,\xx} = \xx + \int_t^s \mu(r, \XX_r^{t,\xx})\D r
+ \int_t^s \Sigma\left(r, \XX_r^{t,\xx}\right)\D W_r,
\end{equation}
where $\mu: [0, T]\times\RR^d\rightarrow\RR^d$
and $\Sigma:[0, T]\times\RR^d\rightarrow\RR^{d\times d}$ adhere to Assumption~\ref{ass:wellposednessBSDE}, 
and~$W$ is a standard $d$-dimensional Brownian motion on the probability space $(\Omega, \Ff, \QQ)$ equipped
with the natural filtration $\FF = \{\Ff_t\}_{0\leq t\leq T}$ of~$W$. 
By the Feynman-Kac formula, options whose discounted expected payoff under~$\QQ$ can be represented as
\begin{equation}
    u(t, \xx)=\EE\left[\int_t^T \E^{-r(s-t)} f\left(s, \XX_s^{t, \xx}\right) \D s + \E^{-r(T-t)} g\left(\XX_T^{t, \xx}\right)\right]\quad \textup{ for all } (t, \xx) \in[0, T] \times \Aa,
\end{equation}
for~$\Aa\subset\RR^d$ with interest rate $r\geq 0$ and continuous functions $f:[0,T]\times\RR^d\to\RR$
and $g:\RR^d\to\RR$ can be viewed as solutions to the Cauchy linear parabolic PDE
\begin{equation*}
\left\{
\begin{array}{rl}
\partial_t u + \Ll u + f - ru= 0, & \text { on }[0, T) \times \Aa, \\
u(T, \cdot)= g, & \text { on } \Aa,
\end{array}
\right.
\end{equation*}
where
\begin{align}\label{eq:infinitesimal_generator}
\Ll u \coloneqq \frac{1}{2}\operatorname{Tr}
\left(\Sigma\Sigma^\top \nabla_{\xx}^2 u\right) +  (\nabla_{\xx} u)\mu,
\qquad\text{on }[0,T)\times \Aa,
\end{align}
is the infinitesimal generator associated with diffusion~\eqref{eq:markovianSDE}.
In this Markovian setting, we thus adopt a setup similar to~\cite{Hure2020} and consider a slightly more general PDE
\begin{equation}\label{eq:markovianPDE}
\left\{
\begin{array}{r@{\;}ll}
    \partial_t u(t,\xx) + \Ll u(t,\xx) + 
    f\Big(t, \xx, u(t,\xx), \nabla_{\xx} u(t,\xx) \cdot \Sigma(t,\xx)\Big) =& 0, & \text{on }[0, T) \times \Aa,\\
     u(T, \cdot) =& g, & \text{on }\Aa,
\end{array}
\right.
\end{equation}
with $f:[0,T]\times \RR^d \times \RR \times \RR^{d} \rightarrow \RR$ 
such that Assumption~\ref{ass:wellposednessBSDE} is satisfied, which guarantees existence and uniqueness of the solution to the corresponding BSDE~\cite[Section 4]{Pardoux1990}.

\begin{assumption}[Well-posedness of the FBSDE system~\eqref{eq:markovianSDE}-\eqref{eq:markovianBSDE}]\label{ass:wellposednessBSDE}
The drift $\mu: [0, T]\times\RR^d\rightarrow\RR^d$ and the diffusion coefficient $\Sigma:[0, T]\times\RR^d\rightarrow\RR^d\times\RR^d$ satisfy global Lipschitz conditions. Moreover,
\begin{enumerate}[(i)]
    \item there exists $L_f>0$ such that 
    $
    \sup_{0 \leq t \leq T}\|f(t, 0,0,0)\|<\infty$
    and, for all $(t_1, x_1, y_1, z_1)$ and $(t_2, x_2, y_2, z_2)$ in $[0, T] \times \RR^d \times \RR \times \RR^{d}$,
    $$
    \left|f\left(t_2, x_2, y_2, z_2\right)-f\left(t_1, x_1, y_1, z_1\right)\right| \leq L_f\left(\sqrt{\left|t_2-t_1\right|}+\left|x_2-x_1\right|+\left|y_2-y_1\right|+\left|z_2-z_1\right|\right),
    $$
    \item The function~$g$ has at most linear growth condition.
\end{enumerate}
\end{assumption}

The corresponding second-order generator is again given by~\eqref{eq:infinitesimal_generator}.
The following assumption is only required to cast the problem into a regression. 
Otherwise, the optimisation (below) can still be solved using other methods, such as stochastic gradient descent. Another solution would be to use the so-called splitting method to linearise the PDE (as in~\cite{Beck2021} and the references therein for example).

\begin{assumption}\label{ass:faffine}
The function $f:[0,T]\times \RR^d \times \RR \times \RR^{d}\times\RR^{d} \rightarrow \RR$ has an affine structure in $\yy \in \RR^m$ and in $\zz,\vv \in \RR^{d}$:
$$
f\left(t, \xx, \yy, \zz, \vv\right)
= a(t,\xx)\yy + b(t, \xx)\zz + c(t, \xx) \vv + \widetilde{f}(t, \xx),
$$
for some real-valued functions $a,b,c,\widetilde{f}$ on $[0,T]\times\RR^d$ that map to conformable dimensions. 
\end{assumption}

\subsection{Random weighted neural network scheme}
The first step in so-called deep BSDE schemes~\cite{Weinan2017, Han2018, Hure2020} is to establish the BSDE associated with the PDE~\eqref{eq:markovianPDE} and the process~\eqref{eq:markovianSDE} through the non-linear Feynman-Kac formula. By~\cite{Pardoux1990} there exist $\FF$-adapted processes $(Y, Z)$, which are unique solutions to the BSDE
\begin{equation}\label{eq:markovianBSDE}
Y_{t}=g\left(\XX_{T}\right)+\int_{t}^{T} f\left(s, \XX_{s}, Y_{s}, Z_{s}\right) \D s-\int_{t}^{T} Z_{s} \D W_{s},
\qquad\text{for any }t \in [0,T],
\end{equation}
and which are connected to the PDE~\eqref{eq:markovianPDE} via
$$
Y_t = u(t, \XX_t) \qquad \text{and} \qquad Z_t = \nabla_x u(t,\XX_t)\cdot\Sigma(t,\XX_t).  
$$
with terminal condition $u(T, \cdot)=g$. Next, the BSDE~\eqref{eq:markovianBSDE} is rewritten in forward form
\begin{equation*}
    Y_{t}=Y_0 - \int_{0}^{t} f\left(s, \XX_{s}, Y_{s}, Z_{s}\right) \D s + \int_{0}^{t} Z_{s} \D W_{s},
\qquad\text{for any }t \in [0,T],
\end{equation*}
and both processes are discretised according to the Euler-Maruyama scheme. To this end let $\pi\coloneqq \left\{0=t_{0}<t_{1}<\ldots<t_{N}=T\right\}$ be a partition of the time interval $[0,T]$ with modulus $|\pi|=\max_{i=\{0,1, \dots, N-1\}} \Deli $ and $\Deli \coloneqq  t_{i+1}-t_{i}$. Then the scheme is given by
\begin{equation}\label{eq:EulerMaruyama}
\left\{
\begin{array}{r@{\;}l}
    \XX_{t_{i+1}} &= \XX_{t_{i}} + \mu(t_{i}, \XX_{t_{i}})\Deli  + \Sigma(t_{i}, \XX_{t_{i}})\DelWi, \\
    Y_{t_{i+1}} &= Y_{t_{i}} - f\left(t_{i}, \XX_{t_{i}}, Y_{t_{i}}, Z_{t_i}\right)\Deli + Z_{t_i}\DelWi,
\end{array}
\right.
\end{equation}
where naturally $\DelWi\coloneqq  W_{t_{i+1}} - W_{t_i}$.
Then for all $i\in\{N-1,\dots,0\}$ we approximate~$u(t_i,\cdot)$ with $\Uf_i(\cdot;\Theta^i)\in\RWNN_{K}^\varrho$ 
and~$Z_{t_i}$ as
\begin{equation*}
\begin{array}{r@{\;}l@{\;}l}
u(t_i,\XX_{t_i}) = Y_{t_i}
& \approx \Uf_i(\XX_{t_i};\Theta^i)
& \coloneqq\Theta^i \Phi_K^i(\XX_{t_i}),\\
Z_{t_i} & \approx \Zf_{i}(\XX_{t_i}) 
& \coloneqq \Dxx \Uf_i(\XX_{t_i};\Theta^i)\cdot\Sigma(t_i,\XX_{t_i})
 = \Theta^i \Dxx \Phi_K^i(\XX_{t_i})\cdot\Sigma(t_i,\XX_{t_i}).
 \end{array}
\end{equation*}
Recall that the derivative~$\Zf_{i}(\XX_{t_i})$ is the approximate derivative from Definition~\ref{def:approximatediff}. The following formulation of the loss function $\ell$ using the approximate derivative is sensible by Lemma~\ref{lem:diffexpectationeq}: notice that for the optimal parameter $\Theta^{i+1, *}$ in step $(i+1)$, the optimal approximation $\widehat \Uf_{i+1}(\XX_{t_{i+1}})\coloneqq\Uf_{i+1}(\XX_{t_{i+1}}; \Theta^{i+1, *})$ 
does not depend on $\Theta^i$, hence under Assumption~\ref{ass:faffine} with $c=0$ the loss function at the $i$-th discretisation step reads
\begin{align*}
    \ell(\Theta^i) &\coloneqq \EE^\Phi\left[\left\| \widehat \Uf_{i+1}(\XX_{t_{i+1}}) - \left[\Uf_i(\XX_{t_i};\Theta^i) - f\Big(t_{i}, \XX_{t_{i}}, \Uf_i(\XX_{t_i};\Theta^i), \Zf_{i}(\XX_{t_i}; \Theta^i)\Big)\Deli + \Zf_{i}(\XX_{t_i}; \Theta^i)\DelWi\right] \right\|^2\right] \\
    &= \EE^\Phi\left[\left\| \widehat \Uf_{i+1}(\XX_{t_{i+1}}) - \left[(\Uf_i(\XX_{t_i};\Theta^i) - \left(a_{i}\Uf_i(\XX_{t_i};\Theta^i) + b_{i}\Zf_{i}(\XX_{t_i}; \Theta^i) + \widetilde{f}_{i}\right)\Deli + \Zf_{i}(\XX_{t_i}; \Theta^i)\DelWi\right] \right\|^2\right] \\
    &= \EE^\Phi\left[\left\| \widehat \Uf_{i+1}(\XX_{t_{i+1}}) + \widetilde{f}_{i}\Deli - \Theta^i\Big\{(1-a_{i}\Deli)\Phi_K^i(\XX_{t_i}) + \Dx \Phi_K^i(\XX_{t_i}) \Sigma_i \left(b_{i}\Deli  + \DelWi\right)\Big\} \right\|^2\right] \\
    &= \EE^\Phi\left[\left\| \Y^i - \Theta^i \X^i \right\|^2\right]
\end{align*}
where $p_{i} \coloneqq  p(t_i, \XX_{t_i})$ 
for $p \in \{a, b, \widetilde{f}, \Sigma\}$, and the expectation $\EE^\Phi$ is of course conditional on the realisation of the random basis $\Phi_K^i$, i.e., conditional on the random weights and biases of the RWNN. Furthermore, we used the notations
\[
\Y^i \coloneqq \widehat \Uf_{i+1}(\XX_{t_{i+1}}) + \widetilde{f}_{i}\Deli
    \qquad\text{and}\qquad
\X^i \coloneqq 
(1-a_{i}\Deli)\Phi_K(\XX_{t_i}) + \Dxx\Phi_K(\XX_{t_i})\cdot \Sigma_i \left(b_{i}\Deli  + \DelWi\right).
\]
The problem can now be solved via least squares from Section~\ref{sec:randomizedLS}, yielding the estimator
$$
\Theta^{i,*} = \EE^\Phi\left[\Y^i\X^{i\top}\right]
\EE^\Phi\left[\X^i\X^{i\top}\right]^{-1}.
$$

\subsection{Algorithm}
We now summarise the algorithmic procedure of our RWNN scheme. See how the algorithm resembles the Least-Square Monte-Carlo method of~\cite{Longstaff2001} after considering sample estimator version of RLS from Section~\ref{sec:randomizedLS}:

\begin{algorithm}[ht]
\caption{RWNN scheme}\label{alg:markovian}
\begin{algorithmic}
\STATE{\textit{Input}:
\begin{ALC@g}
\STATE $\pi=\left\{0=t_{0}<t_{1}<\ldots<t_{N}=T\right\}$ time grid
\end{ALC@g}
}
\STATE{\textit{Initialisation}: 
\begin{ALC@g}
\STATE $\Phi_K^i$ reservoirs with $K\in\NN$ hidden nodes with weights and biases distributed as $\Uu_{[-R,R]}$ with $R>0$ for all $i\in\{0,\dots,N-1\}$
\end{ALC@g}
}
\STATE{\textbf{do}: 
\begin{enumerate}[ ]
    \item Generate $n\in\NN$ paths of
    $\{\XX_{t_i}^{\pi,j}\}_{i=0}^{N}$ for $j\in\{1,\dots,n\}$ with the Euler-Maruyama scheme~\eqref{eq:EulerMaruyama}
    \item Set $\widehat \Uf_N(\xx) = g(\xx)$ for all $\xx\in\RR^d$
\end{enumerate}
}
\FOR{$i\in \{N-1,\dots,0\}$}
\STATE{Approximate $u(t_i, \cdot)$ with $\Uf(\cdot; \Theta^i)\in\RWNN^{\relu}_K$ based on reservoir $\Phi^i_K$}
\STATE{Evaluate the derivative of $\Uf(\cdot; \Theta^i)$ according to~\eqref{eq:RWNN1stderivative}}
\STATE{Solve the regression problem (possibly using the Ridge estimator, see Section~\ref{sec:randomizedLS}) \begin{align*}
\Theta^{i,*} &=\argmin_{\Theta^i} \ell(\Theta^i) = \argmin_{\Theta^i} \EE^{\Phi,n}\left[\left\| \Y^i - \Theta^i\X^i \right\|^2\right]
\end{align*} where \begin{align*}
    \Y^i &\coloneqq \Uf_i(\XX_{t_i};\Theta^i) + \widetilde f\Deli  \\
    \X^i &\coloneqq (1-a)\Phi_K^i(\XX_{t_i}) + (\nabla_x \Phi_K^i(\XX_{t_i})) \Sigma_i \left(b\Deli  + \DelWi\right).\end{align*} and $\EE^{\Phi,n}$ is evaluated over the empirical measure of $\{\XX_{t_i}^{\pi,j}\}_{i=0}^{N}$ for $j\in\{1,\dots,n\}$}
\STATE{Update $\widehat \Uf_i=\Uf_{i}\left(\cdot, \Theta^{i,*}\right)$}
\ENDFOR
\RETURN $\Uf = \{\Uf(\cdot;\Theta^{i,*})\}_{i=0}^{N}$.
\end{algorithmic}
\end{algorithm}

\begin{remark}
We discuss the choice of $R>0$ from Algorithm~\ref{alg:markovian} in different practical scenarios in Section~\ref{sec:RWNN_numerical_results}. We find that the scheme remains robust across different choices of support intervals as long as it aligns with the magnitude of the expected output.
\end{remark}

\newpage
\section{The non-Markovian case}\label{sec:non-Markovian_case}
We now consider a stochastic volatility model under a risk-neutral measure so that $\XX=\nobreak(X, V)$, where the dynamics of log-price process~$X$ are given by
\begin{equation}\label{eq:Xdynamics}
\D X_{s}^{t, x} = \left(r-\frac{V_{s}}{2}\right)\D s
 + \sqrt{V_{s}}\Big(\rho_1 \D W^1_{s} + \rho_2 \D W^2_{s}\Big), \qquad 0 \leq t \leq s \leq T,
\end{equation}
starting from 
$X_{t}^{t, x} =x\in\RR$,
with interest rate $r\in\RR$,
correlation $\rho_1 \in [-1,1]$,
and denote $\rho_2\coloneqq \sqrt{1-\rho_1^{2}}$,
and $W^1, W^2$ are two independent Brownian motions.
We allow for a general variance process process~$V$, satisfying the following:
\begin{assumption}\label{ass:volatilityProcess}
The process~$V$ has continuous trajectories, is non-negative almost surely, adapted to the natural filtration of~$W^1$ and $\EE\left[\int_{0}^{t} V_{s} \D s\right]$ is finite for all $t\geq 0$.
\end{assumption}
By no-arbitrage,
the fair price of a European option with payoff $h:\RR^+\rightarrow\RR^+$ reads
$$
u(t, x) \coloneqq \EE\left[\E^{-r(T-t)} h\left(\E^{X_{T}^{t, x}+r T}\right) \middle\vert \Ff_{t}\right], 
\quad \text{for all }(t, x) \in[0, T] \times \RR,
$$
subject to~\eqref{eq:Xdynamics}.
Since~$\XX$ is not Markovian, one cannot characterise the value function $u(t, x)$ via a deterministic PDE. 
Bayer, Qiu and Yao~\cite{Bayer2022} proved that~$u$ can be viewed as a random field which, together with another random field~$\psi$, satisfies the backward stochastic partial differential equation (BSPDE)
\begin{equation}\label{eq:BSPDEPricing}
-\D u(t,x) = \left[\frac{V_{t}}{2} \partial_x^{2} u(t, x)+\rho \sqrt{V_{t}} \partial_x \psi(t,x)-\frac{V_{t}}{2} \partial_x u(t,x)-r u(t,x)\right] \D t-\psi(t,x) \D W^1_{t},
\end{equation}
in a distributional sense 
for $(t, x) \in[0, T) \times \RR$, with boundary condition
${u(T, x)=h\left(\E^{x+r T}\right)}$ where the variance process $(V_{t})_{t \geq 0}$ is defined exogenously under Assumption~\ref{ass:volatilityProcess}. 
We in fact consider the slightly more general BSPDEs
\begin{equation}\label{eq:BSPDEGeneral}
\begin{array}{r@{\;}l}
-\D u(t,x) &= \displaystyle \bigg\{\frac{V_{t}}{2} \Dx^{2} u(t,x)+\rho \sqrt{V_{t}} \Dx \psi(t,x)-\frac{V_{t}}{2} \Dx u(t,x) \\
& \quad \displaystyle + f\left(t, \E^{x}, u(t,x), \rho_2\sqrt{V_{t}} \Dx u(t,x), \psi(t,x)+ \rho_1 \sqrt{V_{t}} \Dx u(t,x)\right)\bigg\} \D t \\
& \quad \displaystyle -\psi(t,x) \D W^1_{t}, \quad(t, x) \in[0, T) \times \RR, \\
u(T,x) &= \displaystyle g \left(\E^{x}\right), \quad x \in \RR.
\end{array}
\end{equation}
The following assumption on~$f$ and~$g$ (from~\cite{Bayer2022}) ensures well-posedness of the above BSPDE,
and we shall additionally require the existence of a weak-Sobolev solution (Assumption~\ref{ass:RWNNscheme} ) for the convergence analysis of our numerical scheme in Section~\ref{sec:convergence_analysis}. 

\begin{assumption}\label{ass:wellposednessBSPDE}
Let $g:\RR\rightarrow\RR$ and $f:[0,T]\times\RR^4\rightarrow\RR$
be such that
\begin{enumerate}[(i)]
    \item $g$ admits at most linear growth;
    \item $f$ is $L_f$-Lipschitz in all space arguments and there exists $L_0>0$ such that
    $$
    |f(t,x,0,0,0)| \leq L_f(1+|x|) 
        \qquad\text{and}\qquad
        |f(t,x,y,z,\widetilde{z})-f(t,x,y,0,0)| \leq L_0.
        $$\end{enumerate}
\end{assumption}
Note that~\eqref{eq:BSPDEPricing} is just a particular case of the general BSPDE~\eqref{eq:BSPDEGeneral} for the choice $f(t, x, y, z, \widetilde{z}) \equiv -r y$ and $g(\E^{x})\equiv h(\E^{x+r T})$. 
Again, this general form is shown to be well posed in the distributional sense under Assumption~\ref{ass:wellposednessBSPDE} (borrowed from~\cite{Bayer2022}). 
By~\cite{Briand2003} the corresponding BSDE is then,
for $0 \leq t \leq s < T$,
\begin{equation}\label{eq:BSDEGeneral}
\left\{\begin{aligned}
-\D Y_{s}^{t, x} &=f\left(s, \E^{X_{s}^{t, x}}, Y_{s}^{t, x}, {Z_{s}^{1}}^{t, x}, {Z_{s}^{2}}^{t, x}\right) \D s - {Z_{s}^{1}}^{t, x} \D W^1_{s} - {Z_{s}^{2}}^{t, x} \D W^2_{s},\\
Y_{T}^{t, x} &=g\left(\E^{X_{T}^{t, x}}\right),
\end{aligned}\right.
\end{equation}
where $(Y_{s}^{t, x}, {Z_{s}^{1}}^{t, x}, {Z_{s}^{2}}^{t, x})$ is defined as the solution to~\eqref{eq:BSDEGeneral} in the weak sense. 

\subsection{Random neural network scheme}
Let the quadruple $\left(X_{s}, Y_{s}, Z^1_{s}, Z^2_{s}\right)$ be the solution to the forward FBSDE
\begin{equation}\label{eq:logprocess}
\left\{\begin{aligned}
-\D Y_{s} &=f\left(s, \E^{X_{s}}, Y_{s}, Z^1_{s}, Z^2_{s}\right) \D s- Z^1_{s} \D W^1_{s}-Z^2_{s} \D W^2_{s},\\
\D X_{s} &=-\frac{V_{s}}{2} \D s+\sqrt{V_{s}}\left(\rho_1 \D W^1_{s} + \rho_2\D W^2_{s}\right),\\
V_{s} &= \xi_{s} \Ee\left(\eta \widehat{W}_{s}\right), \quad \text { with } \quad \widehat{W}_{s}=\int_{0}^{s} \mathcal{K}(s, r) \D W^1_{r},
\end{aligned}\right.
\end{equation}
for $s \in[0, T)$, with terminal condition
$Y_{T} = g\left(\E^{X_{T}}\right)$, initial condition
$X_{0} = x$ and~$\mathcal{K}$ a locally square-integrable kernel and $\xi_s\coloneqq\xi(s)>0$ is the forward variance curve. 
For notational convenience below, 
we use $\rho_2 \coloneqq  \sqrt{1-\rho_1^2}$, with 
$\rho_1 \in [-1,1]$.
Here~$\Ee(\cdot)$ denotes the Wick stochastic exponential and is defined as $\Ee(\zeta)\coloneqq \exp\left\{\zeta-\frac12\EE[|\zeta|^2]\right\}$ for a centered Gaussian variable $\zeta$. 
Then by~\cite[Theorem~2.4]{Bayer2022},
\begin{align}
Y_{t} &= u\left(t, X_{t}\right),  & \text{ for } t\in[0,T], \\
Z^1_{t} &= \psi\left(t, X_{t}\right) + 
\rho_1 \sqrt{V_{t}} \Dx u\left(t, X_{t}\right), & \textup{ for } t\in[0,T), \\
Z^2_{t} &= \rho_2\sqrt{V_{t}} \Dx u\left(t, X_{t}\right), & \textup{ for } t\in[0,T),
\end{align}
where $(u, \psi)$ is the unique weak solution to~\eqref{eq:BSPDEGeneral}. Accordingly, the forward equation reads
\begin{equation*}
Y_t = Y_0-\int_{0}^{t} f\left(s, \E^{X_{s}}, Y_s, Z_s^1, Z_s^2\right) \D s +\int_{0}^{t}Z_s^1 \D W^1_{s}+\int_{0}^{t} Z_s^2 \D W^2_{s},
\qquad\text{ for }t \in[0, T].
\end{equation*}
By simulating $(W^1, W^2, V)$, 
the forward process~$X$ may be approximated by an Euler scheme--with the same notations as in the Markovian case--and the forward representation above yields the approximation
$$
u\left(t_{i+1}, X_{t_{i+1}}\right) \approx u(t_{i}, X_{t_i}) - f\left(t_i, \E^{X_{t_{i}}}, u\left({t_{i}}, X_{t_{i}}\right), Z^1_{t_i}, Z^2_{t_i}\right) \Deli  + Z^1_{t_i} \Delta^{W^1}_{i} + Z^2_{t_i} \Delta^{W^2}_{i},
$$
with
$$
Z^1_{t_i}  = \rho_1 \sqrt{V_{t_{i}}} \Dx u\left({t_{i}}, X_{t_{i}}\right) + \psi\left({t_{i}}, X_{t_{i}}\right)
\qquad\text{and}\qquad
Z^2_{t_i}  = \rho_2\sqrt{V_{t_{i}}} \Dx u\left({t_{i}}, X_{t_{i}}\right).
$$
By Lemma~\ref{lem:sol_as_RWNN} we can, for each time step $i\in\{0, \dots, N-1\}$, approximate the solutions~$u(t_i,\cdot)$ and~$\psi(t_i,\cdot)$ by two separate networks~$\Uf_i$ and~$\Psi_i$ in $\RWNN_K^{\relu}$: 
\begin{alignat*}{3}
    Y_{t_i} &\approx \Uf_i(X_{t_i}; \Theta^i) && = \Theta^i \Phi^{\Theta,i}_K(X_{t_i}), \\
    Z^1_{t_i} &\approx \Zf^1_i(X_{t_i}; \Theta^i, \Xi^i) && = \Theta^i \left(\Dx \Phi_K^{\Theta,i}(X_{t_i})\right) \rho_1\sqrt{V_{t_i}} + \Xi^i\Phi^{\Xi,i}_K(X_{t_i}), \\
    Z^2_{t_i} &\approx \Zf^2_i(X_{t_i}; \Theta^i, \Xi^i) && = \Theta^i \left(\Dx \Phi_K^{\Theta,i}(X_{t_i})\right) \rho_2\sqrt{V_{t_i}}.
\end{alignat*}
Here $\Phi_K^\Xi$ and $\Phi_K^\Theta$ are realisations of random bases (reservoirs) of the RWNNs with respective parameters~$\Xi$ and~$\Theta$. 
The next part relies on Assumption~\ref{ass:faffine}, namely
$$
f\left(t_i, \E^{X_{t_{i}}},  Y_{t_i},  Z^1_{t_i},  Z^2_{t_i}\right) = a(t_i,X_{t_i}) Y_{t_i} + b(t_i,X_{t_i}) Z^1_{t_i} + c(t_i,X_{t_i}) Z^2_{t_i} + \widetilde{f}(t_i,X_{t_i}),
$$
for some functions $a,b,c,\widetilde{f}$ mapping to~$\RR$, so that, as in the Markovian case, the minimisation of the expected quadratic loss at every time step $i\in\{N-1,\dots,0\}$ reads
\begin{align*}
    &\ell(\Theta^i, \Xi^i) \\ &\coloneqq \EE^\Phi\Bigg[\bigg|\widehat\Uf_{i+1}(X_{t_{i+1}})-\bigg\{\Uf_i(X_{t_i}; \Theta^i)-f\Big(t_{i}, X_{t_{i}}, \Uf_i(X_{t_i}; \Theta^i), \Zf^1_i(X_{t_i}; \Theta^i, \Xi^i), \Zf^2_i(X_{t_i}; \Theta^i, \Xi^i)\Big)\Deli \\
    & \hspace{3cm} + \sum_{k=1}^2\widehat \Zf^k_i(X_{t_i}; \Theta^i, \Xi^i)\Delta^{W^k}_{i}\bigg\}\bigg|^2\Bigg] \\
    &= \EE^\Phi\Bigg[\bigg|\widehat\Uf_{i+1}(X_{t_{i+1}})-\bigg\{\Uf_i(X_{t_i}; \Theta^i)-\left(a\Uf_i(X_{t_i}; \Theta^i) + b\Zf^1_i(X_{t_i}; \Theta^i, \Xi^i) + c\Zf^2_i(X_{t_i}; \Theta^i, \Xi^i) + \widetilde{f}\right)\Deli \\
    & \hspace{3cm} + \sum_{k=1}^2\Zf^k_i(X_{t_i}; \Theta^i, \Xi^i)\Delta^{W^k}_{i}\bigg\}\bigg|^2\Bigg] \\
    &= \EE^\Phi\Bigg[\bigg|\widehat\Uf_{i+1}(X_{t_{i+1}}) + \widetilde{f}\Deli  - \bigg\{ \Xi^i\Phi_K^{\Xi,i}(X_{t_i})\left(\Delta^{W^1}_{i}-b\Deli \right) \\
    & \hspace{3cm} + \Theta^i\left((1-a\Deli )\Phi_K^{\Theta,i}(X_{t_i}) + \Dx\Phi_K^{\Theta,i}(X_{t_i})\sqrt{V_{t_i}}\left(\Delta^{B}_{i} - (b\rho_1 + c\rho_2)\Deli \right) \right) \bigg\}\bigg|^2\Bigg] \\
    &= \EE^\Phi\left[\left|\Y^i - \Xi^i \X_1^i -
    \Theta^i\X_2^i \right|^2\right],
\end{align*}
with $\Delta^{B}_{i}=(\rho_1\Delta^{W^1}_{i} + \rho_2\Delta^{W^2}_{i})$ and where $\widehat\Uf_{i+1}(X_{t_{i+1}})\coloneqq \Uf_{i+1}(X_{t_{i+1}}; \Theta^{i+1,*})$ was set in the previous time step and is now constant (without dependence on~$\Theta^i$). We defined
\begin{equation}\label{eq:SystemNonMark}
\left\{
\begin{array}{r@{\;}l}
\Y^i & \coloneqq  \displaystyle \widehat\Uf_{i+1}(X_{t_{i+1}}) + \widetilde{f}\Deli , \\
     \X_1^i &\coloneqq  \displaystyle \Phi_K^\Xi(X_{t_i})\left(\Delta^{W^1}_{i}-b\Deli \right), \\ 
    \X_2^i &\coloneqq  \displaystyle (1-a\Deli )\Phi_K^\Theta(X_{t_i}) + \Dx\Phi_K^\Theta(X_{t_i})
    \sqrt{V_{t_i}}\left(\Delta^{B}_{i}-(b\rho_1 + c\rho_2)\Deli  \right).
\end{array}
\right.
\end{equation}
In matrix form, this yields
$
\ell(\Theta^i, \Xi^i) = \EE^\Phi[\|\Y^i - \betam^i\X^i\|^2]$,
with 
$\betam^i=\begin{bmatrix}\Xi^i, \Theta^i\end{bmatrix}$
and
$\X^i= \begin{bmatrix}\X_1^i, \X_2^i\end{bmatrix}^{\top}$,
for which the RLS from Section~\ref{sec:randomizedLS}
yields the solution
\begin{equation}\label{eq:betaestimator}
    \betam^i = \EE^\Phi\left[
    \left[\Y^i\X_1^{i\top} \quad \Y^i\X_2^{i\top}\right] \right]
    \EE^\Phi\left[
    \left[\begin{array}{cc}
         \X_1^i\X_1^{i\top} & \X_1^i\X_2^{i\top} \\
         \X_2^i\X_1^{i\top} & \X_2^i\X_2^{i\top}
    \end{array}\right]\right]^{-1}.
\end{equation}


\subsection{Algorithm}
We summarise the steps of the algorithm  below:
\begin{algorithm}[H]\caption{RWNN non-Markovian scheme}\label{alg:nonmarkovian}
\begin{algorithmic}
\STATE{\textit{Inputs}:
time grid $\pi=\left\{0=t_{0}<t_{1}<\ldots<t_{N}=T\right\}$;
number~$K$ of hidden nodes; $R>0$;
}
\STATE{\textit{Initialisation}: 
\begin{ALC@g}
\STATE $\left(\{\Phi_K^{\Theta,i}\}^{N-1}_{i=0}, \{\Phi_K^{\Xi,i}\}^{N-1}_{i=0}\right)$ reservoirs with weights and biases distributed as $\Uu_{[-R,R]}$;
\end{ALC@g}
}
\STATE{\textbf{do}: 
\begin{enumerate}[ ]
    \item Generate $n$ paths of
    $(\{X_{t_i}^{\pi,j}\}_{i=0}^{N}, \{V_{t_i}^{\pi,j}\}_{i=0}^{N})$ for $j\in\{1,\dots,n\}$ with Euler-Maruyama
    \item Set $\widehat{Y}_{N}(x)=g(x)$ for all $x\in\RR$
\end{enumerate}
}
\FOR{$i\in \{N-1,\dots,0\}$}
\STATE{Approximate $u(t_i, \cdot)$ with $\Uf(\cdot; \Theta^i)\in\RWNN^{\relu}_K$ based on reservoir $\Phi^{\Theta,i}_K$}
\STATE{Approximate $\psi(t_i, \cdot)$ with $\Psi(\cdot; \Xi^i)\in\RWNN^{\relu}_K$ based on reservoir $\Phi^{\Xi,i}_K$}
\STATE{Evaluate derivatives of $\left(\Uf(\cdot; \Theta^i), \Psi(\cdot;\Xi^i)\right)$ according to~\eqref{eq:RWNN1stderivative}}
\STATE{Solve the regression problem (possibly using the Ridge estimator from Section~\ref{sec:randomizedLS}) 
$$
\betam^{i,*} = \argmin_{\betam^i} \ell(\betam^i) = \argmin_{\betam^i} \EE^{\Phi,n}\left[\left\|\Y^i - \betam^i\X^i \right\|^2\right]
$$
with $\betam^i=\begin{bmatrix}\Xi^i,\Theta^i\end{bmatrix}$, 
$\X^i= \begin{bmatrix}\X_1^i, \X_2^i\end{bmatrix}^{\top}$,
where $\Y^i, \X_1^i, \X_2^i$ are given in~\eqref{eq:SystemNonMark},
and $\EE^{\Phi,n}$ is computed with the empirical measure of $\left(\{X_{t_i}^{\pi,j}\}_{i=0}^{N}, \{V_{t_i}^{\pi,j}\}_{i=0}^{N}\right)_{j\in\{1,\dots,n\}}$;}
\STATE{Update $\widehat\Uf_{i}(X_{t_{i}})=\Uf_{i}\left(X_{t_i}, \Theta^{i,*}\right)$}
\ENDFOR
\RETURN $\{\Uf_i(\cdot;\Theta^{i,*})\}_{i=0}^{N}$.
\end{algorithmic}
\end{algorithm}

\section{Convergence analysis}\label{sec:convergence_analysis}
In this section, whenever there is any ambiguity, we use the notation~$X^\pi$ to denote the discretised version of the solution process of~\eqref{eq:logprocess} over the partition $\pi=\left\{0=t_{0}<t_{1}<\ldots<t_{N}=T\right\}$ of the interval $[0, T]$,
with modulus $|\pi|=\max_{i\in\{0,1,\dots, N-1\}} \Deli $ with $\Deli =t_{i+1}-t_{i}$. 
As mentioned just before Assumption~\ref{ass:faffine},
the linearity of~$f$ assumed before was only required to cast the optimisation in Algorithm~\ref{alg:nonmarkovian} into a regression problem.
In the forthcoming convergence analysis, this does not play any role, and we therefore allow for a more general function~$f$.

\begin{assumption}\label{ass:RWNNscheme}\
\begin{enumerate}[(i)]
    \item There exists a unique weak solution to the BSPDE system~\eqref{eq:BSPDEGeneral} with $u, \psi \in \Ww^{3,2}$;
    \item There is an increasing continuous function $\omega:\RR^+ \rightarrow\RR^+$ with $\omega(0)=0$ such that
    \[
    \EE\left[\int_{t_1}^{t_2}V_s\D s\right] + \EE\left[\left|\int_{t_1}^{t_2}V_s\D s\right|^2\right]\leq \omega(|t_2-t_1|),
    \quad \text{for any }0\leq t_1 \leq t_2 \leq T;
    \]
    \item There exists $L_f>0$ such that,
    for all $(t,x,z^1,z^2)$ and $(\widetilde{t},\widetilde{x},\widetilde{z}^1,\widetilde{z}^2)$,
    \begin{align}
    & \left|f\left(t, \E^x, y, z^1, z^2\right)
    -f\left(\widetilde{t}, \E^{\widetilde{x}},\widetilde{y},\widetilde{z}^1,\widetilde{z}^2\right)\right| \\
    & \hspace*{2cm} \leq L_f\left\{\omega(|t-\widetilde{t}|)^{\half} + |x-\widetilde{x}| + |y-\widetilde y| + |z^1-\widetilde{z}^1| + |z^2-\widetilde{z}^2|\right\}.
    \end{align}
\end{enumerate}
\end{assumption}

\begin{remark}
Assumption~\ref{ass:RWNNscheme}(iii) may look unusual but appears as we are interested in evaluating options on the stock price, which is the exponential of the log stock price. 
This condition (the same as in~\cite{Bayer2022})
allows us to control the $L^2$ norm of the drift term in~\eqref{eq:logprocess} (see also~\cite[Assumption~2.12(iii)]{bonesini2023rough}) and thus of the backward process~$Y$ therein.
\end{remark}

\begin{assumption}\label{ass:discretisationbound}
Given the partition~$\pi$,
$ \sup_{i\in\{0,\dots, N-1\}}\EE\left[\left|V^\pi_{t_i}\right| \right]$ is finite.
\end{assumption}
{
Classical estimates 
(see Lemma~\ref{lem:X_growth_ineq} for full details and proof) yield}
\begin{equation}
    \EE\left[\sup _{0 \leq t \leq T}\left|X_{t}\right|^{2}\right] \leq C\left(1+\left|x_{0}\right|^{2}\right),\label{eq:solution_moment_bound}
\end{equation}
as well as (Lemma~\ref{lem:X_partition_ineq})
\begin{equation}
    \max_{i\in\{0,\dots,N-1\}} \EE\left[\left|X_{t_{i+1}}-X^\pi_{t_{i+1}}\right|^{2}+\sup _{t \in\left[t_{i}, t_{i+1}\right]}\left|X_{t}-X^\pi_{t_{i}}\right|^{2}\right] \leq C \omega(|\pi|),\label{eq:fwd_process_estimation}
\end{equation}
for some $C>0$ independent of $|\pi|$,
and we furthermore have~\cite{Briand2003}
\begin{equation}\label{eq:L2int_f}
\EE\left[\int_{0}^{T}\left|f\left(t, \E^{X_{t}}, Y_{t}, Z^1_{t}, Z^2_{t}\right)\right|^{2} \D t\right]<\infty,
\end{equation}
as well as the standard $L^{2}$-regularity result on $Y$:
\begin{equation}\label{eq:L2regularity}
\max _{i\in\{0, \dots, N-1\}} \EE\left[\sup _{t \in\left[t_{i}, t_{i+1}\right]}\left|Y_{t}-Y^\pi_{t_{i}}\right|^{2}\right] = \Oo(|\pi|).
\end{equation}
For $k\in\{1,2\}$, define the errors
\begin{equation}\label{eq:Z_error}
\eps^{Z^k}(\pi) \coloneqq \EE\left[\sum_{i=0}^{N-1} \int_{t_{i}}^{t_{i+1}}\left|Z^k_{t}-\overline{Z}^k_{t_{i}}\right|^{2} \D t\right], \quad \text { with } \quad \overline{Z}^k_{t_{i}} \coloneqq \frac{1}{\Deli} \EE_{i}\left[\int_{t_{i}}^{t_{i+1}} Z^k_{t} \D t\right],
\end{equation}
which represent measures of the total variance of the processes~$Z^k$ along the partition~$\pi$ and 
where $\EE_{i}$ denotes the conditional expectation given $\Ff_{t_{i}}$. 
We furthermore define the auxiliary processes, for $i\in\{0, \ldots, {N-1}\}$,
\begin{align}\label{eq:auxilaryProcesses}
\begin{split}
\widehat{\Vv}_{t_{i}} &\coloneqq \EE_{i}\left[\widehat{\Uf}_{i+1}\left(X^\pi_{t_{i+1}}\right)\right]+f\left(t_i, \E^{X^\pi_{t_{i}}}, \widehat{\Vv}_{t_{i}}, \ZowO_{t_{i}}, \ZowT_{t_{i}}\right) \Deli,  \\
\ZowO_{t_{i}} &\coloneqq \widehat\Psi_i(X^\pi_{t_i}) + \frac{1}{\Deli} \EE_{i}\left[\widehat{\Uf}_{i+1}\left(X^\pi_{t_{i+1}}\right) \Delta^{W^1}_{i}\right], \\
\ZowT_{t_{i}} &\coloneqq \frac{1}{\Deli} \EE_{i}\left[\widehat{\Uf}_{i+1}\left(X^\pi_{t_{i+1}}\right) \Delta^{W^2}_{i}\right],
\end{split}
\end{align}
with $\widehat\Uf_i(\xx)\coloneqq\Uf_i(x;\Theta^{i,*})$ and $ \widehat\Psi_i(x)\coloneqq\Psi(x;\Xi^{i,*})$ as before.
Observe that~$\widehat{\Psi}_{i+1}$
and~$\widehat{\Uf}_{i+1}$ do not depend on~$\Theta^i$ 
because the parameters were fixed at $(i+1)$ time step and are held constant at step~$i$ 
(see Algorithm~\ref{alg:nonmarkovian}). Next, notice that $\widehat{\Vv}$ is well defined by a fixed-point argument
since~$f$ is Lipschitz. 
By Assumption~\ref{ass:RWNNscheme}(i), 
there exist $\widehat v_i, \zowk_{t_{i}}$ for which
\begin{equation}\label{eq:existenceAuxilary}
    \widehat\Vv_{t_i} = \widehat v_i(X^\pi_{t_i})\qquad \text{and} \qquad \Zowk_{t_{i}} = \zowk_{t_{i}}(X^\pi_{t_i}) \quad \text{for } i\in\{0,\dots,N-1\}, k\in\{1,2\}.
\end{equation}
By the martingale representation theorem, one has integrable processes~$\ZowO, \ZowT$ such that
\begin{equation}\label{eq:martingaleRep}
    \widehat{\Uf}_{i+1}\left(X^\pi_{t_{i+1}}\right)=\widehat{\Vv}_{t_{i}}-f\left(t_{i},\E^{X^\pi_{t_{i}}}, \widehat{\Vv}_{t_{i}}, \ZowO_{t_{i}}, \ZowT_{t_{i}}\right) \Deli +\int_{t_{i}}^{t_{i+1}}{\widehat{Z}}_{t}^1\D W^1_{t}+\int_{t_{i}}^{t_{i+1}} {\widehat{Z}}_{t}^2 \D W^2_{t},
\end{equation}
since $\Zowk_{t}$ are $\Ff_t^W$-adapted as asserted by the martingale representation theorem. 
From here, It\^o's isometry yields
\begin{align*}
    \ZowO_{t_{i}} &= \widehat\Psi_i(X^\pi_{t_i}) +  \frac1{\Deli}\EE_i\left[\widehat\Uf_{i+1}(X^\pi_{t_{i+1}})\Delta^{W^k}_{i} \right] \\
    &= \frac1{\Deli}\int_{t_i}^{t_{i+1}}\widehat\Psi_i(X^\pi_{t_i})\D t +  \frac1{\Deli}\EE_i\left[\left(\widehat\Vv_{t_i} +\int_{t_{i}}^{t_{i+1}} \widehat{Z}_{t}^1 \D W^1_{t}+\int_{t_{i}}^{t_{i+1}} \widehat{Z}_{t}^2 \D W^2_{t} \right)\int_{t_i}^{t_{i+1}} \D W^1_t \right]  \\
    &= \frac1{\Deli}\EE_i\left[\int_{t_i}^{t_{i+1}}\left( \widehat\Psi_i(X^\pi_{t_i}) + \widehat{Z}_{t}^1 \right)\D t \right],
\end{align*}
and similarly, 
$$
\ZowT_{t_{i}} = \frac1{\Deli}\EE_i\left[\int_{t_i}^{t_{i+1}} \widehat{Z}_{t}^2 \D t \right].
$$
We consider convergence in terms of the following error:
\[
\begin{aligned}
\mathscr{E}\left(\widehat{\Uf}, \widehat\Psi\right) \coloneqq &\max _{i\in\{0, \dots, N-1\}} \EE^\Phi\left[\left|Y_{t_i}-\widehat{\Uf}_i\left(X^\pi_{t_i}\right)\right|^2\right]+\EE^\Phi\left[\sum_{i=0}^{N-1} \int_{t_i}^{t_{i+1}}\sum_{k=1}^2\left|Z_t^k-\widehat{\mathcal{Z}}^k_i\left(X^\pi_{t_i}\right)\right|^2 \D t\right],
\end{aligned}
\]
with $\widehat{\Zz}^1_{i}, \widehat{\Zz}^2_{i}$
introduced before Lemma~\ref{lem:approxError}. 
We now state the main convergence result:
\begin{theorem}\label{thm:convergence}
Under Assumptions~\ref{ass:volatilityProcess}-\ref{ass:wellposednessBSPDE}-\ref{ass:RWNNscheme},
there exists $C>0$ such that
\[
\mathscr{E}\left(\widehat{\Uf}, \widehat\Psi\right) \leq C\left\{\omega(|\pi|) + \EE\left[\left|g(X_T)-g(X^{\pi}_{T})\right|^2\right]+\sum_{k=1}^2\eps^{Z^k}(\pi)+\frac{C^{*}}{K}N+ M|\pi|^2\right\},
\]
with $C^*, M>0$ given in Lemma~\ref{lem:step_RWWN_bound}
and the errors~$\eps^{Z^k}(\pi)$ defined in~\eqref{eq:Z_error}.
\end{theorem}

The following follows from~\eqref{eq:convergence_part2_result}, established in Part~II of the proof of Theorem~\ref{thm:convergence}:

\begin{corollary}\label{cor:convergence}
    Under Assumptions~\ref{ass:volatilityProcess}-\ref{ass:wellposednessBSPDE}-\ref{ass:RWNNscheme},
    there exists $C>0$ such that
    \begin{align}
    \max _{i\in\{0, \dots, N-1\}} & \EE^\Phi\left[\left|Y_{t_i}-\widehat{\Uf}_i\left(X^\pi_{t_i}\right)\right|^2\right] \leq \\ &C\left\{\omega(|\pi|) + \EE\left[\left|g(X_T)-g(X^{\pi}_{T})\right|^2\right]+\sum_{k=1}^2\eps^{Z^k}(\pi)+\frac{C^{*}}{K}N+ M|\pi|^2\right\},
    \end{align}
    with $C^*, M>0$ given in Lemma~\ref{lem:step_RWWN_bound}.
\end{corollary}

\begin{remark}
The second error term is the strong $L^2$-Monte-Carlo error and is $\Oo(N^{-H})$ for processes driven by an fBm with Hurst parameter $H\in(0, 1)$. 
We refer the reader to~\cite{bonesini2023rough, gassiat2023weak} for an exposition on strong versus weak error rates in rough volatility models.
\end{remark}

To prove Theorem~\ref{thm:convergence},
the following bounds on the derivatives are key.
\begin{lemma}\label{lem:derivativeBounds}
Let $\Psi_K(\cdot; \Theta)\in\RWNN^{\relu}_K$ and $(X^\pi_{t_i}, V^\pi_{t_i})_{i}$ denote the discretised versions of~\eqref{eq:logprocess} over the partition~$\pi$, then there exist $L_1, L_2 > 0$ such that, 
for all $i\in\{0,\dots,N-1\}$,
$$
\left\|\EE_i^\Phi\left[\Dx\Psi_K(X^\pi_{t_{i+1}}; \Theta)\right]\right\| \leq L_1 
\qquad \text{and} \qquad \left\|\EE_i^\Phi\left[\Dx^2\Psi_K(X^\pi_{t_{i+1}}; \Theta)\right]\right\| \leq L_2.
$$
\end{lemma}
\begin{proof}
We start with the first derivative. 
For all $x,y\in\RR^d$,
\begin{align*}
    \left\|\Psi_K(x; \Theta)-\Psi_K(y; \Theta)\right\| &= \left\| \Theta \left(\brelu(\Am x+b)-\brelu(Ay+b)\right) \right\| 
    \leq \|\Theta\|_F\|\brelu(\Am x+b)-\brelu(Ay+b)\| \\
    &\leq \|\Theta\|_F\|\Am x-Ay\| \leq \|\Theta\|_F\|A\|_F\|x-y\|
    \leq L_1\|x-y\|,
\end{align*}
since $\relu$ is $1$-Lipschitz. 
The estimator $\Theta$ has an explicit form~\eqref{eq:betaestimator} and its norm is finite, therefore $\Psi_K(\cdot; \Theta)$ is globally Lipschitz and its first derivative is bounded by $L_1>0$. 
Next, without loss of generality, we can set $\Am=\Imat$ and $\bm=0$, since their support is bounded. 
As in~\eqref{eq:SecondDiff}
for $j\in\{1,\dots,m\}$,
\begin{align*}
    \EE_i^\Phi\left[\Dx^2 \Psi_K(X^\pi_{t_{i+1}}; \Theta)\right]_j 
    &= \int \Theta\operatorname{diag}\left(e_{j}\right)\operatorname{diag}\left(\boldsymbol{\delta}_0\left(x-\frac12 V^\pi_{t_i}\Deli  + \sqrt{V_{t_i}}w\right)\right) \boldsymbol p_\Nn(w) \D w \\
    &= \Theta \operatorname{diag}\left(e_{j}\right)\operatorname{diag}\left( \boldsymbol p_\Nn\left(0; x-\frac12 V^\pi_{t_i}\Deli, V^\pi_{t_i}\Deli \right)\right),
\end{align*}
since $\Delta^{B}_{i}\sim\Nn(0, \Deli)$ and $\boldsymbol p_\Nn$ is the Gaussian density applied component-wise. 
Since the weights are sampled on a compact and  $\|\Theta\|$ is finite, then there exists $C>0$ such that
$$
\left\|\EE_i^\Phi\left[\Dx^2 \Psi_K(X^\pi_{t_{i+1}}; \Theta)\right]\right\| \leq C\|\Theta\|_F = L_2.
$$
\end{proof}

From here the error bound of approximating $\widehat\Vv_{t_i}$, $\ZowO_{t_{i}}$ and $\ZowT_{t_{i}}$ with their RWNN approximators $\widehat\Uf_i$, $\widehat{\Zz}^1_{i}$ and $\widehat{\Zz}^2_{i}$ 
(defined in the lemma below)
can be obtained.
For $i\in\{0,\dots,N-1\}$, $(\Uf_i, \Psi_i)\in\RWNN^{\relu}_K$, 
introduce
\begin{equation*}
\begin{array}{r@{\;}lr@{\;}l}
\Zz^1_{i}(x) &\coloneqq \Psi_i(x)+\rho_1\sqrt{V_{t_i}}\Dx\Uf_i(x), & \qquad
\Zz^2_{i}(x) &\coloneqq \rho_2\sqrt{V_{t_i}}\Dx\Uf_i(x),\\
\widehat{\Zz}^1_{i}(x) &\coloneqq \widehat\Psi_i(x)+\rho_1\sqrt{V_{t_i}}\Dx\widehat\Uf_i(x), & \qquad
\widehat{\Zz}^2_{i}(x) &\coloneqq \rho_2\sqrt{V_{t_i}}\Dx\widehat\Uf_i(x).
\end{array}
\end{equation*}

\begin{lemma}\label{lem:approxError}
Under Assumptions~\ref{ass:RWNNscheme}-\ref{ass:discretisationbound}, there exists $M>0$ such that
$$
\EE^\Phi\left[
\left|\Zz_i^k(X^\pi_{t_i}) - \overline{\widehat{Z}}^{k}_{t_{i}}\right|^2\right] \leq\rho_k^2|\pi|^2 M,
\quad
\text{for all }i\in\{0,\dots,N-1\}, k=1,2.
$$
\end{lemma}

\begin{proof}
From~\eqref{eq:auxilaryProcesses} and~\eqref{eq:existenceAuxilary}, we have, for $i\in\{0,\dots,N-1\}$ and $k\in\{1,2\}$,
    \begin{align*}
    \widehat{v}_{i}(x) &= \EE^\Phi_{i}\left[\widehat{\Uf}_{i+1}\left(X^{x,\pi}_{t_{i+1}}\right)\right] + f\left(t_i, \E^{x}, \widehat{v}_{i}(x), \zowO_{i}(x), \zowT_{i}(x)\right) \Deli, \\
    \zowk_{i}(x) &= \Psi_i\left(X^{x,\pi}_{t_{i+1}}\right)\ind_{\{k=1\}} + \frac{1}{\Deli} \EE^\Phi_{i}\left[\widehat{\Uf}_{i+1}\left(X^{x,\pi}_{t_{i+1}}\right) \Delta^{W^k}_{i}\right],
    \end{align*}
    where $X^{x,\pi}_{t_{i+1}} = x + \left(r- \frac12 V_{t_i}\right)\Deli  + \sqrt{V_{t_i}}\Delta^{B}_{i}$ is 
    the Euler discretisation of $\{X_t\}_{t\in[0, T]}$ over~$\pi$ and $\{V^\pi_{t_i}\}_{i=0}^{N}$ is the appropriate discretisation of the volatility process over the same partition. 
    For $\{\Rr^k\}\overset{{\mathrm{iid}}}{\sim}\Nn(0,1)$, the two auxiliary processes can be written as
    \begin{align*}
        \zowk_{i}(x) 
        &= \Psi_i\left(X^{x,\pi}_{t_{i+1}}\right)\ind_{\{k=1\}} + \frac{1}{\Deli} \EE^\Phi_{i}\left[\widehat{\Uf}_{i+1}\left(x - \frac{V^{x,\pi}_{t_i}}{2}\Deli  + \sqrt{V^{x,\pi}_{t_i}\Deli}\left(\rho_1 \Rr^1 + \rho_2 \Rr^2\right)\right) \sqrt{\Deli} \Rr^k\right].
    \end{align*}
    Notice that, while any sensible forward scheme for $\{V_t\}_{t\in[0,T]}$ does depend on a series of Brownian increments, $V^{x,\pi}_{t_i}$ only depends on $\left(\Delta^{W}_{i-1},\dots,\Delta^{W}_{0}\right)$, which are known at time~$t_i$. Thus, since usual derivative operations are available for approximately differentiable functions (Remark~\ref{rem:usualdiffrules}) multivariate integration by parts for Gaussian measures (a formulation of Isserlis' Theorem~\cite{Isserlis1918}) yields
    $$
    \zowk_{i}(x) = \Psi_i\left(X^{x,\pi}_{t_{i+1}}\right)\ind_{\{k=1\}} + \rho_k\sqrt{V_{t_i}}\EE^\Phi\left[\Dx \widehat \Uf_{i+1}\left(X^{x, \pi}_{t_{i+1}}\right)\right],
    $$
    with corresponding derivatives
    $$
    \Dx\zowk_{i}(x) = \Dx\Psi_i\left(X^{x,\pi}_{t_{i+1}}\right)\ind_{\{k=1\}} + \rho_k\sqrt{V_{t_i}}\EE^\Phi\left[\Dx^2 \widehat\Uf_{i+1}\left(X^{x,\pi}_{t_{i+1}}\right)\right].
    $$
An application of the implicit function theorem then implies
\begin{equation}\label{eq:derivativeVhat}
        \Dx \widehat v_i(x) =  \EE^\Phi_{i}\left[\Dx\widehat{\Uf}_{i+1}\left(X^{x,\pi}_{t_{i+1}}\right)\right] + \Deli \left\{\Dx \widehat f_i(x) + \Dr_y \widehat f_i(x)\Dx \widehat v_i(x)
        + \sum_{k=1}^{2}\Dr_{z^k} \widehat f_i(x) \Dx \zowk_{i}(x)\right\},
    \end{equation}
    where $\widehat f_i(x) \coloneqq  f\left(t_i, \E^{x}, \widehat{v}_{i}(x), \zowO_{i}(x), \zowT_{i}(x)\right)$ and
    \begin{align*}
        & \Psi_i\left(X^{x,\pi}_{t_{i+1}}\right)\ind_{\{k=1\}} + \left(1-\Deli \Dr_y\widehat f_i(x)\right)\rho_k\sqrt{V^{\pi}_{t_i}}\Dx \widehat v_i(x) \\
        & = \overline{\widehat{z}}^k_{i}(x)  + \rho_k\sqrt{V^{\pi}_{t_i}}\Deli \left(\Dx \widehat f_i(x) + \Dr_{z^1} \widehat f_i(x) \Dx \zowO_{i}(x) + \Dr_{z^2} \widehat f_i(x)\Dx \zowT_{i}(x)\right).
    \end{align*}
    Thus, for small enough $|\pi|$, 
    \begin{align*}
        &\Psi_i\left(X^{x,\pi}_{t_{i+1}}\right)\ind_{\{k=1\}} + \rho_k\sqrt{V^{\pi}_{t_i}}\Dx \widehat v_i(x) \\ 
        &\qquad \leq \overline{\widehat{z}}^k_{i}(x)  + \rho_k\sqrt{V^{\pi}_{t_i}}\Deli \left(\Dx \widehat f_i(x) + \Dr_{z^1} \widehat f_i(x) \Dx \zowO_{i}(x) + \Dr_{z^2} \widehat f_i(x)\Dx \zowT_{i}(x)\right),
    \end{align*}
    and since~$f$ is Lipschitz by Assumption~\ref{ass:RWNNscheme}(iii), then $\Dr_{r} \widehat f_i(x)=1$ for $r\in\{x, z^1, z^2\}$ and, by Lemma~\ref{lem:derivativeBounds} and the definition of $\overline{\widehat{z}}^k_{i}(x)$:
    \begin{align}
        \Psi_i\left(X^{x,\pi}_{t_{i+1}}\right)\ind_{\{k=1\}} + \rho_k\sqrt{V^{\pi}_{t_i}}\Dx \widehat v_i(x) &\leq \overline{\widehat{z}}^k_{i}(x) + \rho_k\Deli  \sqrt{V^{\pi}_{t_i}}\left(1+\Dx \zowO_{i}(x) + \Dx \zowT_{i}(x)\right) \\
        &\leq \overline{\widehat{z}}^k_{i}(x) + \rho_k\Deli \sqrt{V^{\pi}_{t_i}}\left(1 + L_1 + \sqrt{2} L_2 \sqrt{V^{\pi}_{t_i}}\right).
    \end{align}
Therefore, using the above inequality
\begin{align*}
 & \EE^\Phi\left[\left| \Psi_i(X^\pi_{t_i}) + \rho_1\sqrt{V^\pi_{t_i}}\Dx\Uf_i(X^\pi_{t_i}) - \ZowO_{t_i}\right|^2 \right]  \\
    & \leq 
    \EE^\Phi\left[\left| \zowO_{i}(X^\pi_{t_i}) - \ZowO_{t_i} + \rho_1\Deli \sqrt{V^\pi_{t_i}}\left[1 + L_1 + \sqrt{2} L_2 \sqrt{V^\pi_{t_i}}\right]\right|^2 \right] \\
    & \leq 
    |\rho_1\Deli |^2 \EE\left[\left|  \sqrt{V^\pi_{t_i}}\left(1+L_1+\sqrt{2}L_2\sqrt{V^\pi_{t_i}}\right)\right|^2 \right] \\
    &\leq 
    \rho_1^2|\pi|^2 \left\{\EE\left[\left|V^\pi_{t_i}\right| \right]\left(1 + L_1 + \sqrt{2}L_2\EE\left[\left|V^\pi_{t_i}\right| \right]\right)\right\}
    \leq \rho_1^2|\pi|^2 M,
\end{align*}
relying on Corollary~\ref{coro:diffexpectationeq2} and the fact that~$\Zowk_{t_{i}} = \zowk_{t_{i}}(X^\pi_{t_i})$~(see \eqref{eq:existenceAuxilary}) in the second inequality and the boundedness of~$\EE[|V^{\pi}|]$ from Assumption~\ref{ass:discretisationbound}
in the last line.
The proof of the other bound is analogous.
\end{proof}

\begin{lemma}\label{lem:step_RWWN_bound}
Under Assumptions~\ref{ass:volatilityProcess}-\ref{ass:wellposednessBSPDE}-\ref{ass:RWNNscheme}, for sufficiently small~$|\pi|$ we have 
\begin{equation}
\EE^\Phi\left[\left|\widehat{\Vv}_{t_{i}}-\widehat{\Uf}_{i}\left(X^\pi_{t_{i}}\right)\right|^{2}\right]+\Deli  \EE^\Phi\left[\sum_{k=1}^2\left|\widehat{{Z}}^k_{t_{i}}-\widehat{\mathcal{Z}}^k_{i}\left(X^\pi_{t_{i}}\right)\right|^{2}\right] \leq C\left\{ \frac{C^{*}}{K} +  M|\pi|^3\right\}
\end{equation}
for all $i\in\{0, \dots, N-1\}$ and~$K$ hidden units, for some $C>0$, where~$C^*$ is as in Proposition~\ref{prop:UATRWNN} and~$M$ in Lemma~\ref{lem:approxError}.
\end{lemma}

\begin{proof}[Proof of Lemma~\ref{lem:step_RWWN_bound}]
Fix $i\in\{0,\dots,N-1\}$. Relying on the martingale representation in~\eqref{eq:martingaleRep} and Lemma~\ref{lem:sol_as_RWNN}, we can define the following loss function for the pair $(\Uf_i(\cdot; \Theta),\Psi_i(\cdot;\Xi))\in\RWNN^{\relu}_K$ and their corresponding parameters~$\Theta$ and~$\Xi$:
    \begin{equation}\label{eq:LHat_LTilde}
        \widehat{L}_i(\Theta,\Xi) \coloneqq  \widetilde{L}_i(\Theta,\Xi) + \EE^\Phi\left[\int_{t_i}^{t_{i+1}}\sum_{k=1}^2\left|\widehat Z_t^k - \Zowk_{t_{i}}\right|^2\right],
    \end{equation}
    with
    \begin{align}
        \widetilde{L}_i(\Theta,\Xi) \coloneqq  \EE^\Phi\bigg[\Big\lvert \widehat \Vv_{t_i} - \Uf_i(X^\pi_{t_i}; \Theta) + \Deli  \Big\{ f\left(t_i, \E^{X^\pi_{t_i}}, \Uf(X^\pi_{t_i}; \Theta), \Zz^1_i(X^\pi_{t_i}; \Theta, \Xi), \Zz^2_i(X^\pi_{t_i}; \Theta, \Xi)\right) \\ - f\left(t_i, \E^{X^\pi_{t_i}}, \Vv_{t_i}, \ZowO_{t_{i}}, \ZowT_{t_{i}}\right) \Big\} \Big\rvert^2\bigg] + \Deli \sum_{k=1}^2\EE^\Phi\left[ \left\lvert \Zowk_{t_{i}} - \Zz_i^k(X^\pi_{t_i}; \Theta, \Xi)\right\rvert^2\right].
    \end{align}
    Now, recall the following useful identity, valid for any
    $a, b\in\RR$:
    \begin{equation}
    (a+b)^2 \leq \left(1+\chi\right)a^2 + \left(1+\frac{1}{\chi}\right)b^2, \quad \chi > 0.\label{eq:Young}
    \end{equation}
    
    Applying~\eqref{eq:Young} yields
    \begin{align*}
        & \widetilde{L}_i(\Theta, \Xi) \leq \Deli \sum_{k=1}^2\EE^\Phi\left[ \left\lvert \Zowk_{t_{i}} - \Zz_i^k(X^\pi_{t_i}; \Theta, \Xi)\right\rvert^2\right] + (1 + C\Deli ) \EE^\Phi\left[\left\lvert \widehat \Vv_{t_i} - \Uf_i(X^\pi_{t_i}; \Theta)\right\rvert^2\right] \\ 
        & \quad + \left(1+\frac1{C\Deli }\right)\EE^\Phi\left[\left\lvert  f\left(t_i, \E^{X^\pi_{t_i}}, \Uf(X^\pi_{t_i}; \Theta), \Zz^1_i(X^\pi_{t_i}; \Theta, \Xi), \Zz^2_i(X^\pi_{t_i}; \Theta, \Xi)\right) - f\left(t_i, \E^{X^\pi_{t_i}}, \Vv_{t_i}, \ZowO_{t_{i}}, \ZowT_{t_{i}}\right) \right\rvert^2\right].
    \end{align*}
    Now by the Lipschitz condition on~$f$ from Assumption~\ref{ass:RWNNscheme},
    \begin{equation*}
        \widetilde{L}_i(\Theta, \Xi) \leq (1 + C\Deli ) \EE^\Phi\left[\left\lvert \widehat \Vv_{t_i} - \Uf_i(X^\pi_{t_i}; \Theta)\right\rvert^2\right] + C\Deli \sum_{k=1}^2\EE^\Phi\left[ \left\lvert \Zowk_{t_{i}} - \Zz_i^k(X^\pi_{t_i}; \Theta, \Xi)\right\rvert^2\right].
    \end{equation*}
    For any $a, b\in\RR$, inequality~\eqref{eq:Young} holds with the reverse sign and hence also
    \begin{equation}
        (a+b)^2 \geq \left(1-\chi\right)a^2 - \frac{1}{\chi}b^2, \quad \chi > 0. \label{eq:nYoung}
    \end{equation}
    Following~\eqref{eq:nYoung}, for $\chi=\gamma\Deli$ and $\gamma>0$ we have
    \begin{align*}
        & \widetilde{L}_i(\Theta, \Xi) \geq \Deli \sum_{k=1}^2\EE^\Phi\left[ \left\lvert \Zowk_{t_{i}} - \Zz_i^k(X^\pi_{t_i}; \Theta, \Xi)\right\rvert^2\right] + (1 - \gamma\Deli ) \EE^\Phi\left[\left\lvert \widehat \Vv_{t_i} - \Uf_i(X^\pi_{t_i}; \Theta)\right\rvert^2\right] \\
        & \quad - \frac1{\gamma\Deli }\EE^\Phi\left[\left\lvert  f\Big(t_i, \E^{X^\pi_{t_i}}, \Uf(X^\pi_{t_i}; \Theta), \Zz^1_i(X^\pi_{t_i}; \Theta, \Xi), \Zz^2_i(X^\pi_{t_i}; \Theta, \Xi)\Big) - f\left(t_i, \E^{X^\pi_{t_i}}, \Vv_{t_i}, \ZowO_{t_{i}}, \ZowT_{t_{i}}\right) \right\rvert^2\right].
    \end{align*}
    Again since~$f$ is Lipschitz, the arithmetic-geometric inequality implies
    \begin{align*}
        \widetilde{L}_i(\Theta, \Xi) 
        &\geq (1 - \gamma\Deli ) \EE^\Phi\left[\left\lvert \widehat \Vv_{t_i} - \Uf_i(X^\pi_{t_i}; \Theta)\right\rvert^2\right] + \Deli\sum_{k=1}^2\EE^\Phi\left[ \left\lvert \Zowk_{t_{i}} - \Zz_i^k(X^\pi_{t_i}; \Theta, \Xi)\right\rvert^2\right]\\
        & \qquad - \frac{\Deli}\gamma\EE^\Phi\left[ L_f^2 \bigg\lvert \left\lvert \widehat\Vv_{t_i} - \Uf(X^\pi_{t_i}; \Theta)\right\rvert  + \left\lvert \Zowk_{t_{i}} - \Zz_i^k(X^\pi_{t_i}; \Theta, \Xi)\right\rvert\bigg\rvert^2\right] \\ 
        &\geq (1 - \gamma\Deli ) \EE^\Phi\left[\left\lvert \widehat \Vv_{t_i} - \Uf_i(X^\pi_{t_i}; \Theta)\right\rvert^2\right] + \Deli\sum_{k=1}^2\EE^\Phi\left[ \left\lvert \Zowk_{t_{i}} - \Zz_i^k(X^\pi_{t_i}; \Theta, \Xi)\right\rvert^2\right]\\
        & \qquad - \frac{3\Deli  L_f^2}{\gamma}\left( \EE^\Phi\left[\left\lvert \widehat \Vv_{t_i} - \Uf_i(X^\pi_{t_i}; \Theta)\right\rvert^2\right] + \sum_{k=1}^2\EE^\Phi\left[ \left\lvert \Zowk_{t_{i}} - \Zz_i^k(X^\pi_{t_i}; \Theta, \Xi)\right\rvert^2\right] \right).
    \end{align*}
    Taking $\gamma=6L_f^2$ gives
    \[
        \widetilde{L}_i(\Theta, \Xi) \geq (1-C\Deli )\EE^\Phi\left[\left\lvert \widehat \Vv_{t_i} - \Uf_i(X^\pi_{t_i}; \Theta)\right\rvert^2\right] + \frac{\Deli}2\sum_{k=1}^2\EE^\Phi\left[ \left\lvert \Zowk_{t_{i}} - \Zz_i^k(X^\pi_{t_i}; \Theta, \Xi)\right\rvert^2\right].
    \]
For a given $i\in\{0,\dots,N-1\}$ take $(\Theta^*, \Xi^*)\in\argmin_{\Theta, \Xi} \widehat{L}_i(\Theta, \Xi)$ so that $\widehat\Uf_i=\Uf_i(\cdot; \Theta^*)$ and $\widehat{\Zz}_i^k(\cdot) \coloneqq  \Zz^k_i(\cdot; \Theta^*, \Xi^*)$. 
    From~\eqref{eq:LHat_LTilde}, $\widehat{L}_i$ and~$\widetilde{L}_i$ have the same minimisers, thus combining both bounds gives for all $(\Theta,\Xi)\in\RR^{m\times K}\times\RR^{m\times K}$,
    \begin{align*}
            &\left(1-C \Deli \right) \EE^\Phi\left[\left|\widehat \Vv_{t_{i}}-\widehat\Uf_{i}\left(X^\pi_{t_{i}}\right)\right|^{2}\right]+\frac{\Deli}2\sum_{k=1}^2\EE^\Phi\left[ \left\lvert \Zowk_{t_{i}} - \widehat{\Zz}_i^k\right\rvert^2\right] \leq \widetilde{L}_i(\Theta^*,\Xi^*)\leq \widetilde{L}_{i}(\Theta, \Xi) \\
            &\leq\left(1+C \Deli \right) \EE^\Phi\left[\left|\widehat \Vv_{t_{i}}-\Uf_{i}\left(X^\pi_{t_{i}} ; \Theta\right)\right|^{2}\right]+C \Deli \sum_{k=1}^2\EE^\Phi\left[ \left\lvert \Zowk_{t_{i}} - \Zz_i^k(X^\pi_{t_i}; \Theta, \Xi)\right\rvert^2\right].
    \end{align*}
    Letting $|\pi|$ be sufficiently small gives together with Lemma~\ref{lem:approxError}
    \begin{align*}
        \EE^\Phi\left[\left|\widehat \Vv_{t_{i}}-\widehat\Uf_{i}\left(X^\pi_{t_{i}}\right)\right|^{2}\right] &+ \Deli  \sum_{k=1}^2\EE^\Phi\left[ \left\lvert \Zowk_{t_{i}} - \Zz_i^k(X^\pi_{t_i}; \Theta, \Xi)\right\rvert^2\right] \\ 
        &\leq C \left\{\inf_\Theta\EE^\Phi\left[\left|\widehat v_i(X^\pi_{t_i})-\Uf_i(X^\pi_{t_i}; \Theta)\right|^2\right]
         + |\pi|^3\left(\rho_1^2  + \rho_2^2\right)M\right\},
    \end{align*}
    therefore, using Proposition~\ref{prop:UATRWNN}, we obtain
    \begin{align*}
&\EE^\Phi\left[\left|\widehat{\Vv}_{t_{i}}-\widehat{\Uf}_{i}\left(X^\pi_{t_{i}}\right)\right|^{2}\right] + \Deli  \EE^\Phi\left[\sum_{k=1}^2\left|\Zowk_{t_{i}}-\widehat{\Zz}^k_{i}\left(X^\pi_{t_{i}}\right)\right|^{2}\right] \\
    & \qquad \leq C\left\{ \inf_\Theta\EE^\Phi\left[\left|\widehat v_i(X^\pi_{t_i})-\Uf_i(X^\pi_{t_i}; \Theta)\right|^2\right]
    + M|\pi|^3 \right\}
    \leq C\left\{ \frac{C^{*}}{K}  + M|\pi|^3 \right\}.
    \end{align*}
\end{proof}
The rest of the proof is similar to those in~\cite[Theorem A.2]{Bayer2022} and~\cite[Theorem~4.1]{Hure2020}, but we include it in Appendix~\ref{apx:technical_proofs} for completeness. 

\section{Numerical results}\label{sec:RWNN_numerical_results}
We now showcase the performance of the RWNN scheme on a representative model from each--Markovian and non-Markovian--class. 
We first test our scheme in the multidimensional Black-Scholes (BS) setting~\cite{Black1973} and then move to the non-Markovian setup with the rough Bergomi (rBergomi) model~\cite{Bayer2015}. 
We develop numerical approximations to European option prices given in~\eqref{eq:markovianPDE} and~\eqref{eq:BSPDEPricing}, choosing
\[
f(t, x, y, z^1, z^2) = -r y
\qquad \text {and} \qquad 
g_{\mathrm{call}}\left(x\right) = \left(\E^x-\mathscr{K}\right)^{+},
\]
and discretising over the partition $\pi=\{0=t_0, t_1, \dots t_N=T\}$ for some~$N\in\NN$. 
The precise discretisation schemes of the individual processes are given in their corresponding sections below. 
We remark, however, that the approximated option price for a given Monte-Carlo sample can become (slightly) negative by construction so we add an absorption feature for both models:
\begin{equation}\label{eq:absorption_scheme}
Y^\pi_{t_i} \coloneqq  \max\left\{0,\widetilde Y^\pi_{t_i}\right\}, \qquad \text{for } i\in\{0, \dots, N-1\},
\end{equation}
where $\big\{\widetilde Y^\pi_{t_i}\big\}_{i=0}^N$ denotes the approximation obtained through the RWNN scheme.

\begin{remark}\label{rem:abs_scheme}
    This is a well-studied problem and is especially prevalent in the simulation of square-root diffusions. We acknowledge that the absorption scheme possibly creates additional bias (see~\cite{Lord2009} for the case of the Heston model), however, a theoretical study in the case of the PDE-RWWN scheme is out of the scope of this paper.
\end{remark}

The reservoir used as a random basis of RWNNs here is the classical linear reservoir from Definition~\ref{def:RWNN}. 
For numerical purposes, we introduce a so-called \textit{connectivity} parameter, a measure of how interconnected the neurons in a network are:
the higher the connectivity, the more inter-dependence between the neurons (see~\cite{Dale2020} for effects of connectivity in different reservoir topologies). 
In practice, however, too high a connectivity can lead to overfitting and poor generalisation. 
Recall that our reservoir is given by
\[
\Phi_K: \RR^d \rightarrow \RR^K, \qquad x\mapsto\Phi_K(x)\coloneqq\varrhob(\Am x+\bm),
\]
where only $\Am\in\RR^{K\times d}$ is affected by the connectivity parameter $c\in(0,1]$. 
Mathematically, $\Am_{ij} = \tilde\Am_{ij}\ind_{\{Z_{ij} < c\}}$, where $Z_{ij} \overset{\mathrm{iid}}{\sim} \Uu_{[0,1]}$ and $\tilde\Am_{ij}$ is the original matrix not impacted by the connectivity parameter. A value $c=1$ means that~$\Am$ is dense and fully determined by sampled weights. We find that the choice $c\approx 0.5$ results in superior performance.

In all our experiments, the reservoir weights are sampled uniformly over $[-0.5, 0.5]$ 
(i.e. with the choice $R=0.5$ in Algorithms~\ref{alg:markovian} and~\ref{alg:nonmarkovian}).
We experimented with sampling over alternative distributions and/or intervals, yet discovered that the scheme remains robust to the choice of support, provided it aligns with the magnitude of the expected output.
All experiments below were run on a standard laptop with an 
\texttt{AMD~Ryzen~9~5900HX}
processor without any use of GPU, which would most certainly speed up the algorithms further. The code for both models is available at \href{https://github.com/ZuricZ/RWNN_PDE_solver}{\texttt{ZuricZ/RWNN\_PDE\_solver}}.

\subsection{Example: Black-Scholes}
The Black-Scholes model~\cite{Black1973} is ubiquitous in mathematical finance, allowing for closed-form pricing and hedging of many financial contracts.
Despite its well-known limitations, it remains a reference and is the first model to check before exploring more sophisticated ones. 
Since it offers an analytical pricing formula as a benchmark for numerical results, it will serve as a proof of concept for our numerical scheme.
Under the pricing measure~$\QQ$, the underlying assets $\Sb = (S^1,\ldots, S^d)$ satisfy
\begin{equation}
    \D S^j_t = S^j_t \left(r \D t + \sigma_j \D W^j_t\right), \quad \text{for } t\in[0,T],\; j\in\{1,\dots,d\},
\end{equation}
where $\{W^j_t\}_{t\in[0,T]}$ are standard Brownian motions such that $\langle W^i, W^j \rangle = \rho_{i,j}\D t$, 
with $\rho_{i,j}\in [-1,1]$, $r\geq 0$ is the risk-free rate and $\sigma_j>0$ is the volatility coefficient. 
The corresponding $d$-dimensional option pricing PDE is then given by
\begin{equation}
\frac{\partial u(t,\Sb)}{\partial t}+\sum_{j=1}^d r S^j \frac{\partial u(t,\Sb)}{\partial S^j} + \sum_{j=1}^d \frac{(\sigma_i S^j)^2}{2} \frac{\partial^2 u(t,\Sb)}{(\partial S^j)^2} + \sum_{j=1}^{d-1} \sum_{k=j+1}^d \rho_{j, k} \sigma_j \sigma_k S^j S^k \frac{\partial^2 u(t,\Sb)}{\partial S^j \partial S^k} = r u(t,\Sb),
\end{equation}
for $t\in[0,T)$ with terminal condition 
$u(T,\Sb_T) = g(S^1_T,\dots, S^d_T)$.
To use Algorithm~\ref{alg:markovian}, the process~$\Sb$ has to be discretised, for example with an Euler-Maruyama scheme, for each $j=1,\ldots, d$ and $i\in\{0,1,\dots,N-1\}$:
\begin{equation*}
\left\{
\begin{array}{ll}
    X^{\pi, j}_{t_{i+1}} &= X^{\pi, j}_{t_{i}}
    + \left(r-\frac{\sigma_j^2}{2}\right)\Deli  + \sigma_j\Delta^{W^j}_{i}, \\
    S^{\pi, j}_{t_{i+1}} &= \exp\left\{X^{\pi, j}_{t_{i+1}}\right\},
\end{array}
\right.
\end{equation*}
with initial value $X^{\pi, j}_0=\log\big(S^{\pi, j}_0\big)$.
If not stated otherwise, we let $\mathscr{K}=S_0=1$, $r=0.01$, $T=1$,
and run the scheme with $N=21$ discretisation steps and $n_{\mathrm{MC}}=50,000$ Monte-Carlo samples. 
The reservoir has $K\in\{10, 100, 1000\}$ hidden nodes, in Sections~\ref{sec:results_basket_BS} and~\ref{sec:resutls_call_BS} the connectivity parameter is set to $c=0.5$.

\subsubsection{Convergence rate}\label{sec:BS_convergence_rate}
We empirically analyse the error rate in terms of the number of hidden nodes $K$ obtained in Corollary~\ref{cor:convergence}. To isolate the dependence on the number of nodes, we fix the discretisation grid and the number of MC samples. 
We then consider a single ATM vanilla Call, fix $c=1$, $\sigma=0.1$ and vary
\[
K \in \left\{ \floor{10^{1 + \frac{2(i - 1)}{9}}} : i \in 1, \dots, 10 \right\},
\]
over a set of 10 logarithmically spaced points between 10 and 1000. Due to our vectorised implementation of the algorithm, the reservoir basis tensor cannot fit into the random-access memory of a standard laptop for $K \geq 10000$. The results in Figure~\ref{fig:BSconvergence} are compared to the theoretical price only computed using the Black-Scholes pricing formula. The absorption scheme~\eqref{eq:absorption_scheme} is applied.

\begin{figure}[hbt!]%
    \centering
    \includegraphics[width=0.75\textwidth]{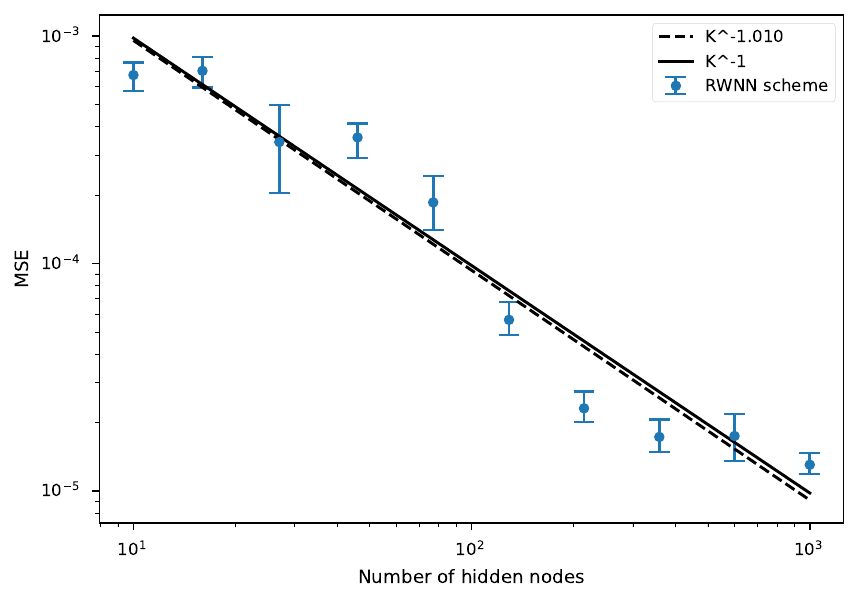}
    \caption{Empirical convergence of the MSE from Corollary~\ref{cor:convergence} under Black-Scholes in terms of the number of hidden nodes (for a fixed grid and number of MC samples). Error bars mark 0.1 and 0.9 quantiles of~20 separate runs of the algorithm. The slope coefficient of the dashed line is obtained through regression of the means of individual runs, while the solid line represents $1/K$ convergence and is shown as a reference.}
    \label{fig:BSconvergence}
\end{figure}

\subsubsection{ATM Call option}\label{sec:resutls_call_BS}
As a proof of concept we first test Algorithm~\ref{alg:markovian} with Call options written on $d\in\NN$ independent assets, i.e. with $\rho_{j,k}=0$ for $j\neq k$ and $V_T = \big(g_{\mathrm{call}}(S^j_T)\big)_{j\in\{1,\dots,d\}}$. This is only done so that the results can be directly compared to the theoretical price computed using the Black-Scholes pricing formula and is, in effect, the same as pricing~$d$ options on~$d$  independent assets, each with their own volatility parameter~$\sigma$. All the listed results in this section are for $K=100$ hidden nodes.
In Table~\ref{tab:BS_d5}, results and relative errors are shown for $d=5$ and $\boldsymbol\sigma\coloneqq(\sigma_1,\dots,\sigma_d)$ uniformly spaced over~$[0.05, 0.4]$. Next, the effects of the absorption scheme~\eqref{eq:absorption_scheme} are investigated. Curiously, the absorption performs noticeably worse compared to the basic scheme, where one does not adjust for negative paths. This leads us to believe that absorption adds a substantial bias, similar to the Heston case (see Remark~\ref{rem:abs_scheme}). Therefore, such a scheme should only be used for purposes, when positivity of the option price paths is strictly necessary (e.g. when hedging). Finally, in Table~\ref{tab:BS_MSE_time}, total MSE and computational times are given for different dimensions. The computational times for different dimensions are then plotted in Figure~\ref{fig:BS_time}. It is important to note that our results do not allow us to make definitive claims about the computational times of the PDE-RWNN scheme across different dimensions. This was not the goal of our experiments, and further detailed theoretical study and experiments would be necessary to draw more definitive conclusions regarding the efficiency of the scheme in various dimensions.

\begin{center}
\begin{table}[hbt!]
\centering
\begin{tabular}{l|cccc}
\multicolumn{5}{c}{Price} \\
\hline
$\boldsymbol\sigma$ & True & PDE (w/ abs) & PDE (w/o abs)  & MC    \\
\hline
0.05 & 0.02521640 & 0.02960256  &  0.02531131  &  0.02574731  \\
0.1  & 0.04485236 &  0.05523114 &  0.04467687  &  0.04547565  \\
0.15 & 0.06459483 &  0.07719949 &  0.06477605  &  0.06520783  \\
0.2  & 0.08433319 &  0.10307868 &  0.08443957  &  0.08484961  \\
0.25 & 0.10403539 &  0.12660871 &  0.10412393  &  0.10513928 
\end{tabular}
\begin{tabular}{l|ccc}
\multicolumn{4}{c}{Rel. Error} \\
\hline
 $\boldsymbol\sigma$ & PDE(w/ abs) & PDE (w/o abs) & MC    \\
\hline
0.05 & 1.74e-01 & -3.76e-03 & -2.11e-02 \\
0.1  & 2.31e-01 & ~3.91e-03 & -1.39e-02 \\
0.15 & 1.95e-01 & -2.81e-03 &	-9.49e-03 \\
0.2  & 2.22e-01 & -1.26e-03 &	-6.12e-03 \\
0.25 & 2.17e-01 & -8.51e-04 &	-1.06e-02
\end{tabular}
\caption{A single run for $d=5$ independent underlyings, where European Calls are compared to the price obtained through PDE-RWNN (\textit{with} and \textit{without} absorption) and the Monte Carlo methods along each dimension. Below, the relative errors of both methods are given. The MC method was run using the same paths as in the PDE-RWNN.}\label{tab:BS_d5}
\end{table}
\end{center}

\begin{center}
\begin{table}[hbt!]
\begin{tabular}{l|cc}
$d$  & Total MSE (with abs) & CPU Time (seconds) \\
\hline
5   & 3.482e-8 & 10.5        \\
10  & 5.417e-8 & 16.0        \\
25  & 4.901e-8 & 34.5        \\
50  & 1.653e-7 & 65.0       \\
100 & 2.534e-7 & 135.0         
\end{tabular}
\caption{Total MSE of the option price calculated across all $d$ assets and CPU training times for varying dimension $d$, where~$\boldsymbol\sigma$ uniformly spaced over $[0.05, 0.4]$.}\label{tab:BS_MSE_time}
\end{table}
\end{center}

\begin{figure}[hbt!]
    \centering
    \includegraphics[scale=0.5]{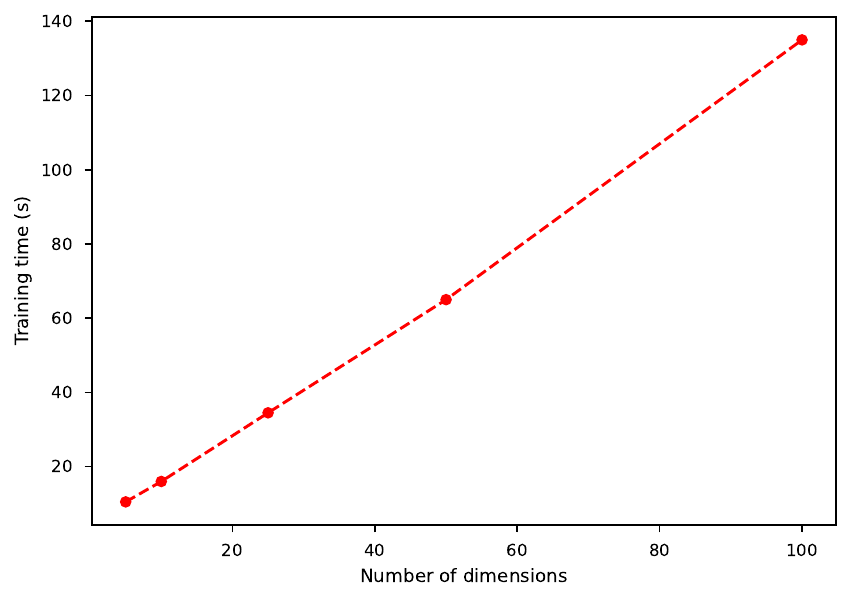}
    \caption{Computational time vs number of dimensions, as in Table~\ref{tab:BS_MSE_time}.}
    \label{fig:BS_time}
\end{figure}

\subsubsection{Computational time}
{A key advantage of RWNNs is their fast training procedure, which in essence relies on solving a linear regression problem. 
We now look at~\eqref{eq:RLS_estimator} to assess its computational complexity. First, computing the sum of outer products, $\sum_{j=1}^n Y_j X_j^\top$, where each $Y_j \in \RR^d$ and $X_j \in \RR^k$, requires forming $n$ matrices of size $d \times k$, resulting in a total computational cost of $\Oo(ndk)$. Similarly, for the sum $\sum_{j=1}^n X_j X_j^\top$, each outer product produces a matrix of size $k \times k$, yielding a cost of $\Oo(nk^2)$. Multiplying the resulting matrices requires $\Oo(dk^2)$ operations, and inversion of the $k \times k$ matrix incurs a cost of $\Oo(k^3)$. The overall computational complexity is therefore
\[
\Oo(ndk + nk^2 + k^3),
\]
where, for large $n$, the sum over $n$ outer products is typically the dominant term. The higher-order terms $\Oo(nk^2)$ and $\Oo(k^3)$ only become significant when $k$ is no longer negligible in comparison to $n$ and $d$. The complexity estimate is consistent with the empirical results presented in Figure~\ref{fig:BS_time}, which demonstrate a linear relationship between the number of dimensions $d$ and observed CPU time. In effect, the dominant $\Oo(ndk)$ term shows the scheme mitigates the curse of dimensionality, allowing high-dimensional problems to be tackled efficiently within this framework.
}



\subsubsection{Basket option}\label{sec:results_basket_BS}
We consider an equally weighted basket Call option with a payoff
\[
g_{\textup{basket}}(\Sb_T)\coloneqq \left(\frac1d\sum_{j=1}^d S^j_T - \mathscr K\right)^+,
\]
where $\mathscr K>0$ denotes the strike price. For simplicity, we consider $d=5$ and an ATM option with $\mathscr K\coloneqq \frac1d\sum_{j=1}^d S^j_0$ and set all $S_0^j=1$ for $j\in\{1,\dots,5\}$. The volatilities $\sigma_j$ are uniformly spaced between~$[0.05, 0.25]$ and the correlation matrix is randomly chosen as
 \[
 \boldsymbol \rho \coloneqq
\begin{bmatrix}
    1  & 0.84  & -0.51 & -0.70 & 0.15  \\
0.84  & 1  & -0.66 & -0.85 & 0.41  \\
-0.51 & -0.66 & 1  & 0.55  & -0.82 \\
-0.70 & -0.85 & 0.55  & 1  & -0.51 \\
0.15  & 0.41  & -0.82 & -0.51 & 1
\end{bmatrix},
 \]
so that $\Sigma\coloneqq \operatorname{diag}(\boldsymbol \sigma)\boldsymbol \rho \operatorname{diag}(\boldsymbol \sigma)$. Since the distribution of a sum of Lognormal is not known explicitly, no closed-form expression is available for the option price. Hence, the reference price is computed using Monte-Carlo with $100$ time steps and $400,000$ samples. In Table~\ref{tab:BS_basket_errors}, we compare our scheme with a classical MC estimator in terms of relative error for $K = 100$ hidden nodes. 

\begin{table}[hbt!]
    \centering
    \begin{tabular}{l|cccc}
        & Reference & PDE (with abs) & PDE (without abs) & MC \\
        \hline
        Price & 0.01624 & 0.01822 & 0.01613 & 0.01625 \\
        \hline
        Rel. error & - & -1.22e-01 & -6.71e-03 & -6.50e-04 \\
        \hline
        Time (seconds) & 12.8 & 9.7 & 9.8 & 0.3
    \end{tabular}
    \caption{Comparison of prices, relative errors and CPU time of the Monte-Carlo estimator, PDE-RWNN scheme \textit{with} and \textit{without} absorption (using same sampled MC paths and $K = 100$) and the reference price computed with 100 time steps and 400,000 samples.}
    \label{tab:BS_basket_errors}
\end{table}

\subsection{Example: Rough Bergomi}
The rough Bergomi model belongs to the recently developed class of rough stochastic volatility models, first proposed in~\cite{Bayer2015, Gatheral2018, guennoun2018asymptotic}, 
where the instantaneous variance is driven by a fractional Brownian motion (or more generally a continuous Gaussian process) with 
Hurst parameter $H < \frac{1}{2}$. 
As highlighted in many papers, they are able to capture many features of (Equities, Commodities,...) data,
and clearly seem to outperform most classical models, with fewer parameters.
Precise examples with real data can be found in~\cite{Bayer2015} for SPX options, in~\cite{gatheral2020quadratic, horvath2020volatility} for joint SPX-VIX options and in~\cite{bennedsen2022decoupling, Gatheral2018} for estimation on historical time series, the latter  being the state-of-the-art statistical analysis under the $\PP$-measure.
We consider here the price dynamics under~$\QQ$ with constant initial forward variance curve~$\xi_0(t)>0$ for all $t\in[0,T]$:
\begin{equation*}
\left\{
\begin{array}{r@{\;}l}
\displaystyle\frac{\D S_t}{S_t} &= r \D t + \sqrt{V_t} \D \left(\rho\D W^1_t + \sqrt{1-\rho^2} W^2_t\right), \\
V_t &= \displaystyle \xi_0(t)\Ee\left( \eta\sqrt{2H} \int_0^t (t-u)^{H-\frac{1}{2}} \D W_u^1 \right),
\end{array}
\right.
\end{equation*}
where $\eta>0$, $\rho \in (-1,1)$ and $H\in(0,1)$. 
The corresponding BSPDE reads
\begin{equation}\label{eq:roughBergomi}
-\D u(t,x) = \left[\frac{V_{t}}{2} \partial_x^{2} u(t, x)+\rho \sqrt{V_{t}} \partial_x \psi(t,x)-\frac{V_{t}}{2} \partial_x u(t,x)-r u(t,x)\right] \D t-\psi(t,x) \D W^1_{t},
\end{equation}
with terminal condition $u(T, x) = g_{\mathrm{call}}\left(\E^{x+r T}\right)$. 
While the existence of the solution  was only proven in the distributional sense~\cite{Bayer2022}, we nevertheless apply our RWNN~scheme.
To test Algorithm~\ref{alg:nonmarkovian}, both the price and the volatility processes are discretised according to the Hybrid scheme developed in~\cite{Bennedsen2017, McCrickerd2018}. 
We set the rBergomi parameters as $(H, \eta, \rho, r, T, S_0) = (0.3, 1.9, -0.7, 0.01, 1, 1)$ and choose the forward variance curve to be flat with $\xi_0(\cdot) = 0.235^2$.
Again, we are pricing an ATM vanilla Call option with $\mathscr K=S_0=1$. 
The number of discretisation steps is again $N=21$, the number of Monte-Carlo samples $n_{\mathrm{MC}}=50,000$ and the reservoir has $K\in\{10, 100, 1000\}$ nodes with connectivity $c=0.5$ in Section~\ref{sec:resutls_call_rB}. 

\subsubsection{Convergence rate}\label{sec:rB_convergence_rate}

As in Section~\ref{sec:BS_convergence_rate}, we conduct an empirical analysis of the convergence error from Corollary~\ref{cor:convergence} for the same ATM Call. 
To isolate the dependence on the number of nodes we fix $c=1$, $n_{\mathrm{MC}}=50,000$ and vary
\[
K \in \left\{ \floor{10^{1 + \frac{2(i - 1)}{9}}} : i \in 1, \dots, 10 \right\},
\]
logarithmically spaced points between 10 and 1000.
The reference price is computed by Monte-Carlo with~$100$ time steps and $800,000$ samples. 
The absorption scheme has been applied and the results are displayed in Figure~\ref{fig:rBconvergence}. 
In this section, the same random seed was used as in Section~\ref{sec:BS_convergence_rate}, to ensure consistent results across different simulations.

\begin{figure}[hbt!]
    \centering
    \includegraphics[width=0.75\textwidth]{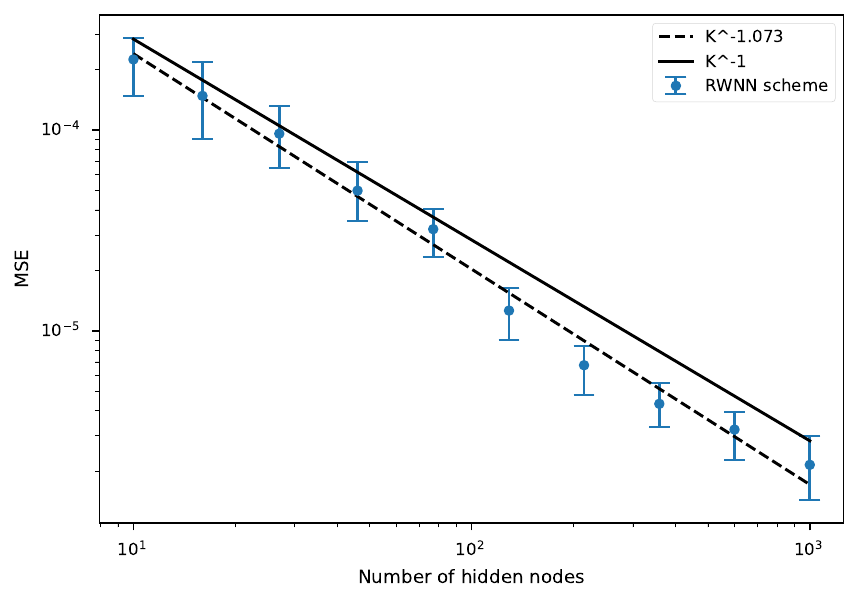}
    \caption{Empirical convergence of MSE under rBergomi in terms of  the number of hidden nodes. Error bars mark 0.1 and 0.9 quantiles of 20 separate runs of the algorithm. The slope coefficient of the dashed line is obtained through regression of the means of individual runs, while the solid line represents $1/K$ convergence and is shown as a reference.}
    \label{fig:rBconvergence}
\end{figure}

\subsubsection{ATM Call option}\label{sec:resutls_call_rB}

We now evaluate the performance of our PDE-RWNN method for option pricing in the rough Bergomi model using the market parameters listed above and compare the results to those obtained with the MC method over the same sample paths. We also investigate the effect of the absorption scheme (Table~\ref{tab:rB_vanilla_errors}) and find that, interestingly, despite keeping the paths positive, the absorption scheme adds noticeable bias. 
Nevertheless, the relative error of the proposed scheme with absorption 
is comparable to the results using regular artificial networks found in the literature~\cite[Table~1]{Bayer2022}, yet, our scheme learns much faster with orders of magnitudes faster training times.

\begin{table}[hbt!]
    \centering
    \begin{tabular}{l|c|c|c|c}
        & Reference & PDE (with abs) & PDE (without abs) & MC \\
        \hline
        Price & 0.07993 & 0.081924 & 0.07973 & 0.080310 \\
        \hline
        Rel. error & - & 24.9e-03 & 2.54e-03 & -4.73e-03 \\
        \hline
        Time (seconds) & 10.1 & 7.4 & 7.5 & 0.4
    \end{tabular}
    \caption{Prices, relative errors and CPU time of the Monte-Carlo estimator, PDE-RWNN scheme \textit{with} absorption, PDE-RWNN scheme \textit{without} absorption, both with $K = 100$ and same sampled MC paths. Reference price computed with 100 time steps and 800,000 samples.}
    \label{tab:rB_vanilla_errors}
\end{table}

\newpage
\bibliographystyle{siam}
\bibliography{biblio} 

\appendix

\section{Error bounds for RWNN}\label{apx:errorRWNN}

\begin{proposition}[Proposition 3 in~\cite{gonon2023approximation}]\label{prop:UATRWNN}
Suppose $\psi^{*}: \RR^{q} \rightarrow \RR$ can be represented as
$$
\psi^{*}(z)=\int_{\RR^{q}} \E^{\I\theta^\top z} g(\theta) \D \theta,
$$
for some complex-valued function~$g$ on $\RR^{q}$ and all $z \in \RR^{q}$ with $\|z\| \leq Q$. Assume that
$$
\int_{\RR^{q}} \max \left(1,\|\theta \|^{2 q+6}\right)|g(\theta)|^{2} \D \theta <\infty .
$$
For $R>0$, suppose the rows of the $K\times K$-valued random matrix~$\Am$ are iid Uniform on $\Bb_{R} \subset \RR^{q}$, that the entries of $\bm\in\RR^{K}$ are iid Uniform on $[-\max (Q R, 1), \max (Q R, 1)]$, 
that~$\Am$ and~$\bm$ are independent and let $\relu(x)\coloneqq \max (x, 0)$ on~$\RR$. 
Let the pair $(A,b)$ characterises a random basis (reservoir)~$\Phi$ and its corresponding network $\Psi\in\RWNN^{\relu}_K$ in the sense of Definition~\ref{def:RWNN}. 
Then, there exists an $\RR^K$-valued random variable~$\Theta$ and $C^{*}>0$ 
(given explicitly in~\cite[Equation~(33)]{gonon2023approximation}) such that
$$
\EE^\Phi\left[\left\|\Psi(Z;\Theta) - \psi^{*}(Z)\right\|^{2}\right] \leq \frac{C^{*}}{K},
$$
and for any $\delta \in(0,1)$, with probability $1-\delta$ the random neural network $\Psi(\cdot;\Theta)$ satisfies
$$
\left(\int_{\RR^{q}}\left\|\Psi(z;\Theta) - \psi^{*}(z)\right\|^{2} \mu_{Z}(\D z)\right)^{1 / 2} \leq \frac{\sqrt{C^{*}}}{\delta \sqrt{K}}.
$$
\end{proposition}

\begin{lemma}\label{lem:sol_as_RWNN}
For any $t_i\in\pi$, 
there exists~$g$ as in Proposition~\ref{prop:UATRWNN}
such that solutions $\mathfrak{f}\in\{u(t_i, \cdot), \psi(t_i, \cdot)\}$ to the BSPDE~\eqref{eq:BSPDEGeneral} can be represented as
$$
\mathfrak{f}(z)=\int_{\RR^q} \E^{\I\theta^\top z} g(\theta) \D \theta,
\qquad\text{ for all }z\in\RR^q.
$$
\end{lemma}
\begin{proof}

A sufficient condition for this is that $\mathfrak{f} \in L^{1}\left(\mathbb{R}^{q}\right)$
has an integrable Fourier transform and belongs to the Sobolev space $\Ww^{2,2}\left(\mathbb{R}^{q}\right)$~\cite[Theorem~6.1]{Folland1995}. 
In our case, for $q=1$ and with $\mathfrak{f}\in \Ww^{3,2}(\RR^1)$ as enforced by Assumption~\ref{ass:RWNNscheme}~(i), \cite[Theorem~9.17]{Folland2013} implies that $\|\widehat{(\Dr^{\boldsymbol\alpha} \mathfrak{f})}\|_{(1)} \leq C \|\mathfrak{f}\|_{(3)}$ for all multi-indices $|\boldsymbol\alpha|\leq 2$ and $\|\mathfrak f\|_{(s)}=[\int|\widehat{\mathfrak f}(\xi)|^2\left(1+|\xi|^2\right)^s \D \xi]^{1 / 2}$ for~$s\in\RR$ denotes the Sobolev norm. Thus $\widehat{ \mathfrak{f}} \in L^1(\RR^1)$
by~\cite[Corollary~2.52]{Folland2013}.
\end{proof}

\section{Proof of Theorem~\ref{thm:convergence}}
\label{apx:technical_proofs}
{\small We prove the theorem in three steps. 
First, we derive an estimate for the $L^2$-distance of~$\widehat{\Vv}_{t_i}$ to the discretised~$Y_{t_i}$.
Next, we use this to estimate the $Y$-component and finally for the $Z$-components. 
In the first part, for convenience, we do not explicitly denote the conditionality of expectation on the realisation of random bases of the corresponding RWNNs, but remark that the expectation should be understood as such whenever this is the case. For convenience, we introduce the error notations:
\begin{equation}\label{eq:ErrorNotations}
\begin{array}{r@{\;}ll}
\Ef\left[A, B\right] & \coloneqq \EE\left[\left|A-B\right|^2\right], \\
    \Df_i\left[A, B\right] & \coloneqq \EE_i\left[\left|A - B\right|\right], \\
    \If_i\left[A(\cdot), B(t_i)\right] & \displaystyle \coloneqq \EE\left[\int_{t_i}^{t_{i+1}}\left|A(t) - B(t_i)\right|^2 \D t \right],
\end{array}
\end{equation}
where the error rates~$\Df$ and~$\If$ are implicitly related to the partition~$\pi$ through the index $i\in\{0,\dots,N-1\}$.
We further denote
$f_{t} \coloneqq  f\left(t, X_t, Y_t, Z^1_t, Z^2_t\right)$
and note that the constant $C>0$ might change from line to line.

\paragraph{\textbf{Part I}} We start by showing for each $i\in\{0, \dots, N-1\}$,
\begin{equation}\label{eq:convergence_part1_result}
\ErrYiVhat
\leq (1+C|\pi|) \ErrYiUfip
 + C |\pi| \EE\left[\int_{t_i}^{t_{i+1}}f_{t}^2 \D t\right] 
 + C \sum_{k=1}^{2}\ErrIntZkibar + C |\pi|\omega(|\pi|),
\end{equation}
where $\omega$ is the modulus function from Assumption~\ref{ass:RWNNscheme}. 
By writing the SPDE as the corresponding BSPDE as in~\eqref{eq:BSDEGeneral} and using~\eqref{eq:auxilaryProcesses}, we obtain
\[
Y_{t_i}-\widehat{\Vv}_{t_i} = 
\EE_i\left[Y_{t_{i+1}}-\widehat{\Uf}_{i+1}\right]+\EE_i\left[\int_{t_i}^{t_{i+1}} f\left(t, \E^{X_t}, Y_t, Z^1_t, Z^2_t\right)-f\left(t_i, \E^{X^\pi_{t_i}}, \widehat{\Vv}_{t_i}, \ZowO_{t_{i}}, \ZowT_{t_{i}}\right) \D t\right].
\]
Then, Young's inequality~\eqref{eq:Young} with $\chi = \gamma\Deli$ gives
\begin{equation}
\begin{aligned}
&\ErrYiVhat \\
& \leq \EE\left\{\left(1+\gamma \Deli \right)
\ErriYiUfip^{2}
 + \left(1+\frac{1}{\gamma \Deli }\right)\EE_i\left[\int_{t_i}^{t_{i+1}}\left\{
f_{t}-f\left(t_i, \E^{X^\pi_{t_i}}, \widehat{\Vv}_{t_i}, \ZowO_{t_{i}}, \ZowT_{t_{i}}\right)\right\} \D t\right]^2\right\},
\end{aligned}
\end{equation}
and Cauchy-Schwarz, Assumption~\ref{ass:RWNNscheme} and~\eqref{eq:fwd_process_estimation} imply
\begin{align}
\begin{split}
\ErrYiVhat &\leq \left(1+\gamma \Deli \right) \EE\left[\ErriYiUfip^2\right] \\ 
& + 5\left(1+\frac{1}{\gamma \Deli }\right) L_f^2 \Deli \Bigg\{C |\pi| \omega(|\pi|)
+ \ErrIntYtVihat + \sum_{k=1}^{2}\ErrIntZkihat\Bigg\}.
\end{split}
\end{align}
The standard inequality $(a+b)^2\leq 2(|a|^2+|b|^2)$ and the $L^2$-regularity of~$Y$ in~\eqref{eq:L2regularity} yield
\begin{align*}
\ErrIntYtVihat 
 & = \EE\left[\int_{t_i}^{t_{i+1}}|Y_t - \widehat \Vv_{t_i}|^2\D t\right] = \EE\left[\int_{t_i}^{t_{i+1}}|Y_t-Y_{t_i}+Y_{t_i}-\widehat\Vv_{t_i}|^2 \D t\right]\\
  & \leq 2 \EE\left[\int_{t_i}^{t_{i+1}}
\left\{|Y_t-Y_{t_i}|^2 + \left|Y_{t_i}-\widehat\Vv_{t_i}\right|^2\right\} \D t\right] \\
 & \leq 2|\pi|^2 + 2\Deli\EE\left[|Y_{t_i}-\widehat\Vv_{t_i}|^2\right]
 = 2|\pi|^2 + 2\Deli\ErrYiVhat,
\end{align*}
so that after rearranging the constant term $\left(1+\frac{1}{\gamma\Deli}\right)L_f^2\Deli=(1+\gamma\Deli)\frac{L_f^2}\gamma$,
we obtain
\begin{equation}\label{eq:convergence_part1_0}
\begin{aligned}
\ErrYiVhat &\leq \left(1+\gamma \Deli \right) \EE\left[\ErriYiUfip^{2}\right] \\
 & \qquad + 5\left(1+\gamma \Deli \right) \frac{L_f^2}{\gamma}\left\{
C |\pi| \omega(|\pi|) + 2\Deli  \ErrYiVhat
+ \sum_{k=1}^{2}\ErrIntZkihat\right\}.
\end{aligned}
\end{equation}
Since $\overline{Z}^k$ are defined as $L^2$-projections of $Z$ (see~\eqref{eq:Z_error}), the last term reads
\begin{align}\label{eq:convergence_part1_1}
\begin{split}
& \ErrYiVhat \leq\left(1+\gamma \Deli \right) \EE\left[\ErriYiUfip^2\right] \\
 & \qquad + 5\left(1+\gamma \Deli \right) \frac{L_f^2}{\gamma}\Bigg\{C |\pi| \omega(|\pi|) + 2 \Deli  \ErrYiVhat 
+ \sum_{k=1}^2\left(\ErrIntZkibar + \Deli \ErrZkibarhat \right)\Bigg\}.
\end{split}
\end{align}
The rightmost term can be further expanded:
integrating the BSDE~\eqref{eq:BSDEGeneral} over $[t_i, t_{i+1}]$, multiplying it by $\Delta^{W^k}_{i}$ for $k\in\{1,2\}$ separately and using the definitions in~\eqref{eq:auxilaryProcesses} give
\begin{align}
\Deli \left[\overline{Z}^k_{t_i}-\Zowk_{t_{i}}\right] &= \EE_i\left[\Delta^{W^k}_{i}\left\{Y_{t_{i+1}}-\widehat{\Uf}_{i+1}-\EE_i\left[Y_{t_{i+1}}-\widehat{\Uf}_{i+1}\right]\right\}\right] 
+ \EE_i\left[\Delta^{W^k}_{i} \int_{t_i}^{t_{i+1}} f_{t}\D t\right].
\end{align}
Next, taking the expectation of the square and using the H{\"o}lder inequality yield
\begin{align*}
    & \frac{\Deli^2}{2}
    \EE_i\left[\left|\overline{Z}^k_{t_i} - \Zowk_{t_{i}}\right|^2\right] \\
    &\leq \left|\EE_i\left[\Delta^{W^k}_{i}\left\{Y_{t_{i+1}}-\widehat{\Uf}_{i+1}-\EE_i\left[Y_{t_{i+1}} - \widehat{\Uf}_{i+1}\right]\right\}\right]\right|^2 + \left|\EE_i\left[\Delta^{W^k}_{i} \int_{t_i}^{t_{i+1}} f_{t}\D t\right]\right|^2 \\
    &\leq \EE_i\left[\left|\Delta^{W^k}_{i}\right|^2\right] \EE_i\left[\left| Y_{t_{i+1}}-\widehat{\Uf}_{i+1} - \EE_i\left[Y_{t_{i+1}} - \widehat{\Uf}_{i+1}\right]\right|^2\right] + \EE_i\left[ \left|\Delta^{W^k}_{i}\right|^2\right]\EE_i\left[ \left|\int_{t_i}^{t_{i+1}} f_{t}\D t\right|^2\right] \\
    &= \Deli \left\{\EE_i\left[\left| Y_{t_{i+1}}-\widehat{\Uf}_{i+1}\right|^2\right] -2\EE_i\left[\left(Y_{t_{i+1}}-\widehat{\Uf}_{i+1}\right)\EE_i\left[Y_{t_{i+1}}-\widehat{\Uf}_{i+1}\right]\right] + \left|\EE_i\left[ Y_{t_{i+1}}-\widehat{\Uf}_{i+1}\right]\right|^2  + \EE_i\left[ \left|\int_{t_i}^{t_{i+1}} f_{t}\D t\right|^2\right]\right\} \\
    &\leq \Deli \left\{\EE_i\left[\left| Y_{t_{i+1}}-\widehat{\Uf}_{i+1}\right|^2\right] - \left|\EE_i\left[ Y_{t_{i+1}}-\widehat{\Uf}_{i+1}\right]\right|^2 + \EE_i\left[\int_{t_i}^{t_{i+1}}\D t  \int_{t_i}^{t_{i+1}} \left| f_{t}\right|^2\D t\right]\right\}.
\end{align*}
Finally, by the law of iterated conditional expectations
\begin{equation}\label{eq:convergence_part1_2}
\frac{\Deli}{2}
\EE\left[\left|\overline{Z}^k_{t_i} - \Zowk_{t_{i}}\right|^2\right]
\leq \ErrYiUfip
- \EE\left[\ErriYiUfip^{2}\right]
 + \Deli  \EE\left[\int_{t_i}^{t_{i+1}}
f_{t}^2 \D t\right],
\end{equation}
which can then be used in~\eqref{eq:convergence_part1_1} to obtain
\begin{align}
\begin{split}
\ErrYiVhat
& \leq\left(1+\gamma \Deli \right) \EE\left[\ErriYiUfip^2\right] 
 + 5\left(1+\gamma \Deli \right) \frac{L_f^2}{\gamma}\Bigg\{C |\pi| \omega(|\pi|) + 2\Deli  \ErrYiVhat
 +\sum_{k=1}^2\ErrIntZkibar\\
&\qquad + 4\left(\ErrYiUfip - \EE\left[\ErriYiUfip^{2}\right]\right)
 + 4 \Deli  \EE\left[\int_{t_i}^{t_{i+1}}f_{t}^2 \D t\right]\Bigg\} \\
& \leq [1+20 L_f^2 \Deli ] \ErrYiUfip
 + C \left\{|\pi| \omega(|\pi|)
 + \Deli \ErrYiVhat
 + \sum_{k=1}^{2}\ErrIntZkibar 
 + \Deli  \EE\left[\int_{t_i}^{t_{i+1}}f_{t}^2 \D t\right]\right\},
\end{split}
\end{align}
with $\gamma = 20 L_f^2$ in the second inequality, concluding Part~I by letting~$|\pi|$ sufficiently small.

\paragraph{\textbf{Part II}} 
We now prove an estimate for the $Y$-component and show that
\begin{align}\label{eq:convergence_part2_result}
\max_{i\in\{0,\dots,N-1\}} 
\ErrYiUfi \leq C\Bigg\{\omega(|\pi|) + \EE\left[\left|g(X_T)-g(X^\pi_T)\right|^2\right]+\sum_{k=1}^2\eps^{Z^k}(\pi)+\frac{C^{*}}{K}N + M|\pi|^2\Bigg\}.
\end{align}
With Young's inequality of the form $(a+b)^2\geq (1-|\pi|)a^2 - \frac1{|\pi|} b^2$,
we have
\begin{equation}\label{eq:convergence_part2_0}
\ErrYiVhat = \EE\left[\left|Y_{t_i} - \widehat{\Uf}_i + \widehat{\Uf}_i - \widehat{\Vv}_{t_i}\right|^2\right] 
\geq (1-|\pi|)\ErrYiUfi - \frac1{|\pi|} \ErrVhatiUfi.
\end{equation}
Since $\frac1{|\pi|}\leq \frac{N}{T}=CN$,
taking small enough~$|\pi|$, 
\begin{equation}\label{eq:convergence_part2_1}
\ErrYiUfi \leq \frac{\ErrYiVhat + CN \ErrVhatiUfi}{(1-|\pi|)}
    \leq C \left\{
\ErrYiVhat + CN \ErrVhatiUfi\right\}
\end{equation}
and combining~\eqref{eq:convergence_part1_result} from Part~I with~\eqref{eq:convergence_part2_1},
\begin{align*}
 \ErrYiUfi & \leq [1+C|\pi|] \ErrYiUfip
 + C\left\{|\pi| \EE\left[\int_{t_i}^{t_{i+1}}f_{t}^2 \D t\right]
  + \sum_{k=1}^{2}\ErrIntZkibar
  + |\pi|\omega(|\pi|) + N\ErrVhatiUfi
  \right\}.
\end{align*}
After noting that $Y_{t_N}= g(X_T)$ and $\widehat{\Uf}_N= g(X^\pi_T)$, recalling the definition of $\eps^{Z^k}(\pi)$ from~\eqref{eq:Z_error} and the $L^2$-integrability of $f$ in~\eqref{eq:L2int_f}, 
straightforward induction implies
\begin{align}
\max_{i\in\{0, \dots, N-1\}}\ErrYiUfi
 &\leq C\Bigg\{\omega(|\pi|) + |\pi| + \EE\left[\left|g(X_t)-g(X^\pi_T)\right|^2\right] + \sum_{k=1}^2\eps^{Z_k}(\pi)
 + N\sum_{i=0}^{N-1}\ErrVhatiUfi\Bigg\}.
\end{align}
Finally, this in conjunction with Lemma~\ref{lem:step_RWWN_bound}, safely ignoring the $Z$-component on the left-hand side, since it is positive, gives the desired result~\eqref{eq:convergence_part2_result}.

\paragraph{\textbf{Part III}} 
Finally, we prove the following bound on the~$Z$ components:
\begin{equation}\label{eq:convergence_part3_result}
 \EE\left[\sum_{i=0}^{N-1} \int_{t_i}^{t_{i+1}}\sum_{k=1}^2\left|Z_t^k-\widehat{\mathcal{Z}}^k_i\left(X^\pi_{t_i}\right)\right|^2 \D t\right] \leq C\Bigg\{\omega(|\pi|) + \EE\left[\left|g(X_T)-g(X^\pi_T)\right|^2\right]+\sum_{k=1}^2\eps^{Z^k}(\pi)+\frac{C^{*}}{K}N + M|\pi|^2\Bigg\}.
\end{equation}
Since~$\overline{Z}^k$ are $L^2$-projections of~$Z$,
\begin{align*}
\ErrIntZkihat 
&= \EE\left[\int_{t_i}^{t_{i+1}}\left|Z^k_t - \overline{Z}^{k}_{t_{i}} + \overline{Z}^{k}_{t_{i}} - \Zowk_{t_{i}}\right|^2\D t\right] \\
 &= \EE\left[ \int_{t_i}^{t_{i+1}}\left\{
 \left|Z^k_t - \overline{Z}^{k}_{t_{i}}\right|^2 + \left|\overline{Z}^{k}_{t_{i}} - \Zowk_{t_{i}}\right|^2 +
 2
 \left|\overline{Z}^{k}_{t_{i}} - \Zowk_{t_{i}}\right|
 \left|Z^k_t - \overline{Z}^{k}_{t_{i}}\right|
 \right\}\D t \right]\\
 &= \ErrIntZkibar + \Deli \ErrZkibarhat + 2\EE\left[\left|\overline{Z}^{k}_{t_{i}} - \Zowk_{t_{i}}\right| \int_{t_i}^{t_{i+1}}\left|Z^k_t - \frac{1}{\Deli} \EE_{i}\left[\int_{t_{i}}^{t_{i+1}} Z^k_{s} \D s\right]\right| \D t \right],
\end{align*}
and the mixed term is cancelled by the tower property. 
Using~\eqref{eq:convergence_part1_2} below yields
\begin{align*}
\ErrIntZkihat
& =\ErrIntZkibar + \Deli \ErrZkibarhat \\
&\leq \ErrIntZkibar +  2 \left\{\ErrYiUfip
 - \EE\left[\ErriYiUfip\right]\right\}
 + 2 |\pi| \EE\left[\int_{t_i}^{t_{i+1}}f_{t}^2 \D t\right],
\end{align*}
for $k\in\{1,2\}$. 
Summing over $i\in\{0,\dots,N-1\}$ together with~\eqref{eq:L2int_f} implies
\begin{align}\label{eq:convergence_part3_0}
\ErrIntZkihat \leq \eps^{Z^k}(\pi) & + 2\EE\left[\left|g(X_T)-g(X^\pi_T)\right|^2\right]\\
&  + 2\sum_{i=0}^{N-1}\left\{\ErrYiUfi - \EE\left[\ErriYiUfip^{2}\right]\right\} + C|\pi|.
\end{align}
The summation index was changed in the first term in the curly braces to apply the terminal conditions $Y_{t_N}\coloneqq g(X_T)$ and $\widehat{\Uf}_N(X^\pi_{t_N})\coloneqq g(X^\pi_T)$.
Rearranging~\eqref{eq:convergence_part2_0} into 
$
\ErrYiUfi
\leq \frac1{|\pi|(1-|\pi|)} \ErrVhatiUfi + \frac1{1-|\pi|}\ErrYiVhat$ and combining it with~\eqref{eq:convergence_part1_0} give
\begin{equation}
\begin{aligned}
2\left\{\ErrYiUfi - \EE\left[\ErriYiUfip^{2}\right]\right\}
&\leq \frac{3 \ErrVhatiUfi}{|\pi|(1-|\pi|)} + \frac{2}{1-|\pi|}\Bigg\{\left(1+\gamma |\pi|\right) \EE\left[\ErriYiUfip^{2}\right] \\ 
& \quad + \frac{5\left(1+\gamma |\pi|\right) L_f^2}{\gamma}\left[C |\pi|\omega(|\pi|) +2|\pi| \ErrYiVhat + \sum_{k=1}^{2}\ErrIntZkihat\right]\Bigg\}.
\end{aligned}
\end{equation}
Plugging this back into~\eqref{eq:convergence_part3_0} gives
\begin{equation}
\begin{aligned}
&\ErrIntZkihat \leq \eps^{Z^k}(\pi) + 2\EE\left[\left|g(X_T)-g(X^\pi_T)\right|^2\right] \\ 
& \qquad + \sum_{i=0}^{N-1} \frac{3}{|\pi|(1-|\pi|)} \ErrVhatiUfi +
\sum_{i=0}^{N-1}\Bigg\{2\frac{1+\gamma |\pi|}{1-|\pi|} \EE\left[\ErriYiUfip\right] \\ 
& \quad + \frac{\left(1+\gamma |\pi|\right)}{1-|\pi|} \frac{10 L_f^2}{\gamma}\left(C \omega(|\pi|)|\pi| +2|\pi| \ErrYiVhat + \sum_{k=1}^{2}\ErrIntZkihat\right)\Bigg\} + C|\pi|.
\end{aligned}
\end{equation}
Furthermore, choose $\gamma=50L_f^2$ so that $\frac{\left(1+\gamma |\pi|\right)}{1-|\pi|} \frac{10 L_f^2}{\gamma}\leq\frac14$ for small $|\pi|$,
\begin{equation}
\begin{aligned}
    \frac{1}{2} \sum_{k=1}^2 \ErrIntZkihat 
    & \leq \sum_{k=1}^2\eps^{Z^k}(\pi) + C \Bigg\{\max _{i\in\{0, \dots, N\}} \ErrYiUfi + \omega(|\pi|) + |\pi| + \EE\left[\left|g\left(X_T\right)-g\left(X^\pi_T\right)\right|^2\right] \\
    & \qquad + |\pi| \sum_{i=0}^{N-1} \ErrYiVhat + N \sum_{i=0}^{N-1} \ErrVhatiUfi \Bigg\},
\end{aligned}
\end{equation}
in conjunction with Part~I~\eqref{eq:convergence_part1_result},
\begin{equation}
\begin{aligned}
    &\frac{1}{2} \sum_{k=1}^2 \ErrIntZkihat 
    \leq \sum_{k=1}^2\eps^{Z^k}(\pi) + C\Bigg\{\max _{i\in\{0, \dots, N\}} \ErrYiUfi + \omega(|\pi|) + |\pi| + \EE\left[\left|g\left(X_T\right)-g\left(X^\pi_T\right)\right|^2\right] \\
    & \qquad + |\pi| \sum_{i=0}^{N-1}\bigg\{C|\pi|\omega(|\pi|) + C \sum_{k=1}^2\ErrIntZkibar + (1+C|\pi|) \ErrYiUfip + |\pi| \EE\left[\int_{t_i}^{t_{i+1}}f_{t}^2 \D t\right]\bigg\} \\
    & \qquad + N \sum_{i=0}^{N-1} \ErrVhatiUfi \Bigg\},
\end{aligned}
\end{equation}
and $L^2$-integrability $f$ in~\eqref{eq:L2int_f}, Lemma~\ref{lem:step_RWWN_bound}, Part~II~\eqref{eq:convergence_part2_result} gives
\begin{equation}
    \frac{1}{2} \sum_{k=1}^2 \ErrIntZkihat \leq C\left\{\sum_{k=1}^2\eps^{Z^k}(\pi) + \omega(|\pi|)+|\pi|+\EE\left[\left|g\left(X_T\right)-g\left(X^\pi_T\right)\right|^2\right] + \frac{C^{*}}{K}N + M|\pi|^2\right\}.\label{eq:convergence_part3_1}
\end{equation}
Ultimately, since~$\overline{Z}$ is an $L^2$-projection, we have the relation, for $k\in\{1,2\}$,
\begin{equation}
\EE\left[ \int_{t_i}^{t_{i+1}}\left|Z_t^k-\widehat{\mathcal{Z}}^k_i\left(X^\pi_{t_i}\right)\right|^2 \D t\right] \leq 2\ErrIntZkihat + 2\Deli  \EE\left[ \int_{t_i}^{t_{i+1}}\left|\Zowk_{t_{i}}-\widehat{\mathcal{Z}}^k_i\left(X^\pi_{t_i}\right)\right|^2 \D t\right].
\end{equation}
Then Lemma~\ref{lem:step_RWWN_bound} applied over~$\Qq$ and summing over $i\in\{0,\dots,N-1\}$ yield Part~III.
}

\section{Technical Lemmata}\label{apx:technical_lemmata}

\begin{lemma}\label{lem:X_growth_ineq}
    Under Assumptions~\ref{ass:volatilityProcess}~and~\ref{ass:RWNNscheme}(ii), the solution to the SDE in~\eqref{eq:Xdynamics} starting at $x_0>0$ satisfies the following growth inequality:
    $$
    \EE\left[\sup_{0 \leq t \leq T} \left|X_t^{0,x_0}\right|^2\right] 
    \leq C\left(1 + |x_0|^2\right),
    \quad\text{for some }C>0.
    $$
\end{lemma}
\begin{proof}
For simplicity, we write $X = X^{0,x_0}$ in the proof.
For any $t\geq 0$, we can decompose~$X_t$ in three parts:
    \[
    X_t = x_0 + \underbrace{\int_0^t \left(r - \frac{V_s}{2}\right) \D s}_{D_t} + \underbrace{\int_0^t \sqrt{V_s}
    \left(\rho_1 \D W_s^1 + \rho_2 \D W_s^2\right)}_{M_t}.
    \]
    We can bound the martingale term~$M$ using the BDG inequality and Assumption~\ref{ass:RWNNscheme}(ii) so that there exists $C_M>0$ for which 
    \[
    \EE\left[\sup_{0 \leq t \leq T} |M_t|^2\right] \leq C_M \EE\left[\int_0^T V_s \D s\right] \leq C_M \omega(T).
    \]
    For the drift term~$D$, we now have
    \[
    \sup_{0 \leq t \leq T} |D_t| \leq rT + \frac{1}{2} \int_0^T V_s \D s.
    \]
    Since $V\geq 0$ almost surely by Assumption~\ref{ass:volatilityProcess}, we can square both sides, interchange the square and supremum, use $(a + b)^2 \leq 2(a^2 + b^2)$ and Assumption~\ref{ass:RWNNscheme}(ii) so that
    $$
    \EE\left[\left\{ \sup_{0 \leq t \leq T} |D_t|\right\}^2\right] = 
    \EE\left[\sup_{0 \leq t \leq T} |D_t|^2\right]
    \leq 
    \EE\left[2r^2T^2 + \frac{1}{2}\left(\int_0^T V_s \D s\right)^2\right]
     \leq
     2r^2T^2 + \frac{1}{2}\omega(T).
    $$
    Combining all terms and using $(a + b + c)^2 \leq 3(a^2 + b^2 + c^2)$ we have
    \begin{align*}
    \EE\left[\sup_{0 \leq t \leq T} 
 |X_t|^2\right]
 & = \EE\left[\sup_{0 \leq t \leq T} 
 \left|x_0 + D_t + M_t\right|^2\right]\\
 & \leq 3\left\{|x_0|^2
 + \EE\left[\sup_{0 \leq t \leq T} 
 |D_t|^2\right] + \EE\left[\sup_{0 \leq t \leq T} |M_t|^2\right]\right\}\\
 & \leq 3|x_0|^2 + 6r^2T^2 + \frac{3}{2}\omega(T) + 3C_M \omega(T),
    \end{align*}
and the lemma follows.
\end{proof}
\begin{lemma}\label{lem:X_partition_ineq}
    Under Assumptions~\ref{ass:RWNNscheme}(ii) and~\ref{ass:discretisationbound}, the solution to~\eqref{eq:Xdynamics} starting at $x_0>0$ satisfies
    $$
    \max _{i \in\{0, \ldots, N-1\}} \EE\left[\left|X_{t_{i+1}}-X_{t_{i+1}}^\pi\right|^2+\sup _{t \in\left[t_i, t_{i+1}\right]}\left|X_t-X_{t_i}^\pi\right|^2\right] \leq C \omega(|\pi|),
    $$
    where $|\pi|$ is the mesh size of the partition $\pi=\left\{t_0, t_1, \ldots, t_N\right\}$.
\end{lemma}
\begin{proof}
    \textbf{Error at grid points: $\left|X_{t_{i+1}}-X_{t_{i+1}}^\pi\right|^2$.}
    The difference between the true solution and its Euler-Maruyama approximation at grid points reads
    \begin{align*}
    X_{t_{i+1}} - X^\pi_{t_{i+1}}
    & = 
    \int_{t_i}^{t_{i+1}} \left(r - \frac{V_s}{2}\right) \D s 
    - \left(r - \frac{V^\pi_{t_i}}{2}\right)(t_{i+1} - t_i) \\
    & + \int_{t_i}^{t_{i+1}} \sqrt{V_s} \D W_s -     \sqrt{V^\pi_{t_i}} \left(W_{t_{i+1}} - W_{t_i}\right),
    \end{align*}
    where $W=\rho_1 W^1 + \rho_2 W^2$. 
    Since $\EE[|V^\pi_{t_i}|]$ is finite for all $t_i \in \pi$ by Assumption~\ref{ass:discretisationbound}, using Assumption~\ref{ass:RWNNscheme} ii) it is easy to see that the drift term contributes an error of order $\Oo(\omega(|\pi|))$.
    For the diffusion term:
    \begin{align*}
    \EE\left[\left|\int_{t_i}^{t_{i+1}} \sqrt{V_s} \D W_s-\sqrt{V_{t_i}}\left(W_{t_{i+1}} - W_{t_i}\right)\right|^2\right]
 & = \EE\left[\left|\int_{t_i}^{t_{i+1}}\left(\sqrt{V_s}-\sqrt{V_{t_i}}\right) \D W_s\right|^2\right]\\
  & =\EE\left[\int_{t_i}^{t_{i+1}}\left|\sqrt{V_s}-\sqrt{V_{t_i}}\right|^2 \D s\right]
    \end{align*}
    by It\^o isometry.
    Since
    $\left|\sqrt{V_s}-\sqrt{V_{t_i}}\right|^2 \leq |V_s-V_{t_i}|$,
    then, by Assumption~\ref{ass:RWNNscheme}(ii) the diffusion term is bounded by $\Oo(\omega(\pi))$.\\
    \noindent
    \textbf{Supremum error over subintervals: $\sup _{t \in\left[t_i, t_{i+1}\right]}\left|X_t-X_{t_i}^\pi\right|^2$}.
    For any subinterval,
    $\Ee_t := X_t - X_{t_i}^\pi$ is the local error due to discretisation. 
    Using standard SDE arguments and properties of the Euler-Maruyama scheme (see for example~\cite{Hutzenthaler2012}), this again yields an upper bound of order $\Oo(\omega(|\pi|))$. 
    Since expectation is linear, the result follows.
    \end{proof}
\end{document}